\documentclass[11pt]{article}
\usepackage{amsmath, amssymb, amsfonts, amsthm}
\usepackage{epsfig, graphicx}
\newtheorem{theorem}{Theorem}[section]
\newtheorem{proposition}[theorem]{Proposition}
\newtheorem{corollary}[theorem]{Corollary}
\newtheorem{lemma}[theorem]{Lemma}

\newcommand{\abs}[1]{\lvert #1 \rvert}

\newcommand{\tri}{| \! | \! |}
\newcommand{\rd}{{\rm d}}
\newcommand{\be}{\begin{equation}}
\newcommand{\ee}{\end{equation}}
\newcommand{\bey}{\begin{eqnarray}}
\newcommand{\eey}{\end{eqnarray}}

\newcommand{\eps}{\varepsilon}

\newcommand{\bp}{{\bf p}}

\newcommand{\bx}{{\bf x}}

\newcommand{\bn}{{\bf n}}

\newcommand{\ph}{\varphi}

\renewcommand{\a}{\alpha}

\newcommand{\cU}{{\cal U}}

\newcommand{\bR}{{\mathbb R}}

\newcommand{\bN}{{\mathbb N}}
\newcommand{\bZ}{{\mathbb Z}}

\newcommand{\tr}{\mbox{Tr}}

\newcommand{\wt}{\widetilde}
\newcommand{\wh}{\widehat}

\newcommand{\cK}{{\cal K}}

\newcommand{\cL}{{\cal L}}

\input epsf

\newcommand{\donothing}[1]{}

\begin{document}

\title{Derivation of the two-dimensional nonlinear Schr\"odinger equation from many body quantum dynamics}
\author{Kay Kirkpatrick\thanks{Massachusetts Institute of Technology, 77 Mass. Ave., Cambridge, MA 02139, USA},
Benjamin Schlein\thanks{Institute for Mathematics, LMU Munich, Theresienstrasse 39, 80333 Munich, Germany},
and  Gigliola Staffilani${}^*$ }

\maketitle
\abstract{
We derive rigorously, for both $\mathbb{R}^2$ and $[-L,L]^{\times 2}$, the cubic nonlinear Schr\"odinger equation in a suitable scaling limit from the two-dimensional many-body Bose systems with short-scale repulsive pair
interactions. We first prove convergence of the solution of the BBGKY hierarchy, corresponding to the many-body systems, to a solution of the infinite Gross-Pitaevskii hierarchy, corresponding to the cubic NLS; and then we prove uniqueness for the infinite hierarchy, which requires number-theoretical techniques in the periodic case.
}

\section{Introduction}
\setcounter{equation}{0}

Bose-Einstein condensation is an usual state of matter near absolute zero, where the particles (bosons) are so supercooled that they all fall into the ground state and exhibit quantum mechanical behavior macroscopically--as described by the cubic nonlinear Schr\"odinger (or, Gross-Pitaevskii) equation \cite{G, P}. Experimental physicists have used Bose-Einstein condensates to make atom lasers and to convert light to matter and back, suggesting potential applications that include more accurate measurements via interferometry as well as quantum information processing. The fragility of this state of matter makes it all the more important to develop the rigorous theory.

\medskip

We consider an $N$-boson system described on the Hilbert space $L^2_s (\Lambda^{N})$, the subspace of $L^2 (\Lambda^{N})$ consisting of permutation symmetric functions, by the Hamiltonian
\begin{equation}\label{eq:ham} H_N = \sum_{j=1}^N -\Delta_{x_j} + \frac{1}{N} \sum_{i<j}^N N^{2\beta} V (N^{\beta} (x_i -x_j)) \, .\end{equation}
The goal of this paper is to investigate this system when $\Lambda= [-L,L]^{\times 2}$, that is the particles are confined in a finite square ($L$ is fixed), and in this case we impose periodic boundary conditions at the boundary of the square. A easier problem is to  consider the situations $\Lambda= \bR^2$, that is the particles can move on the full two-dimensional space and for completeness we will consider this case as well.
In (\ref{eq:ham}), we assume the interaction potential $V$ to be positive, and sufficiently regular, and we will consider $0< \beta <1$.

\medskip

We are interested in the dynamics governed by the $N$-particle Schr\"odinger equation
\begin{equation}\label{eq:schr}
i\partial_t \psi_{N,t} = H_N \psi_{N,t}
\end{equation}
with an asymptotically factorized initial data $\psi_{N,0} = \psi_N \in L^2_s (\Lambda^N)$. Here asymptotic factorization means factorization of the marginal densities associated with the $N$-particle wave function $\psi_N$ in the limit $N \to \infty$.

\medskip

Recall that, for $k =1, \dots, N$, the $k$-particle density matrix associated with an $N$-particle wave function $\psi_N$ is defined as the non-negative trace class operator $\gamma^{(k)}_{N,t}$ on $L^2 (\Lambda^k)$ with kernel given by \[ \gamma^{(k)}_{N} (\bx_k ; \bx'_k) = \int \rd \bx_{N-k} \, \overline{\psi}_N (\bx_k, \bx_{N-k}) \psi_N (\bx'_k, \bx_{N-k}) \] where we use the notation $\bx_k = (x_1, \dots ,x_k), \bx'_k = (x'_1,\dots, x'_k) \in \Lambda^k$, $\bx_{N-k} = (x_{k+1}, \dots , x_N) \in \Lambda^{N-k}$. In other words, $\gamma^{(k)}_{N}$ is defined by taking the partial trace of the orthogonal projection $\gamma_N = |\psi_N \rangle \langle \psi_N |$ over the last $N-k$ particles. Our main result is the following theorem.
\begin{theorem}\label{thm:main}
Suppose that $\Lambda = \bR^2$ or $\Lambda = [-L, L]^{\times 2}$, for some $L >0$. Assume moreover that $V \in W^{2,\infty} (\Lambda)$, $V \geq 0$ and that $0 < \beta <3/4$. Consider a family $\{\psi_N \}_{N \in \bN}$ such that
\begin{itemize}
\item[$\bullet$] $\{ \psi_N \}_{N \in \bN}$ has bounded energy per particle:
\begin{equation}\label{eq:en0}
\sup_{N \in \bN} \; \frac{1}{N} \langle \psi_N , H_N \psi_N \rangle < \infty  \, ,
\end{equation}
\item[$\bullet$] $\{ \psi_N \}_{N \in \bN}$ exhibits asymptotic factorization: there exists $\ph \in L^2 (\Lambda)$ such that
\begin{equation}\label{eq:asyfact} \tr \, \left| \gamma^{(1)}_N - |\ph \rangle \langle \ph| \right| \to 0 \qquad  \text{as } \quad N \to \infty \, . \end{equation} Here $\gamma^{(1)}_N$ is the one-particle density associated with $\psi_N$.
\end{itemize}
Denote by $\psi_{N,t} = e^{-iH_N t} \psi_N$ the solution to the $N$-particle Schr\"odinger equation (\ref{eq:schr}) with initial data $\psi_N$ and let $\gamma^{(k)}_{N,t}$ be the $k$-particle marginal associated with $\psi_{N,t}$. Then we have, for every $t \in \bR$ and $k \geq 1$,
\begin{equation}\label{eq:claim}
\tr \, \left| \gamma^{(k)}_{N,t} - |\ph_t \rangle \langle \ph_t |^{\otimes k} \right| \to 0 \qquad \text{as } \quad N \to \infty  \end{equation}
where $\ph_t$ is the solution to the cubic nonlinear Schr\"odinger equation \begin{equation}\label{eq:nls} i \partial_t \ph_t = - \Delta \ph_t + b_0 |\ph_t|^2 \ph_t \end{equation} with $b_0 = \int_{\Lambda} \rd x V(x)$ and $\ph_{t=0} = \ph$.
\end{theorem}

{\it Remark. } The condition $0 < \beta < 3/4$, which is only used in the proof of Proposition \ref{prop:enest} can be relaxed to $0 < \beta < 1$; this follows from the observation that, for $k=2$, (\ref{eq:Hkest}) holds for all $0< \beta <1$ (this is clear from the proof of Proposition \ref{prop:enest}). The inequality (\ref{eq:Hkest}) for $k > 2$, on the other hand, is only needed to obtain the a-priori bounds (\ref{eq:aprik-inf}) which, however, can also be proven using a modification of the energy estimate (\ref{eq:Hkest}) with appropriate space-cutoffs (in the same spirit as in \cite[Proposition 7.1, Theorem 7.3]{ESY3}); to keep our discussion as simple as possible, we will only discuss the case $0 < \beta < 3/4$.

\medskip

The main idea in the proof of this theorem is to study the time-evolution of the marginal densities $\{ \gamma^{(k)}_{N,t} \}_{k=1}^N$ in the limit $N \to \infty$. Starting from the Schr\"odinger equation (\ref{eq:schr}) it is simple to verify that the marginal densities satisfy a hierarchy of $N$ coupled equations, commonly known as the BBGKY hierarchy
\begin{equation}\label{eq:BBGKY1}
\begin{split}
i\partial_t \gamma^{(k)}_{N,t} = \; &\sum_{j=1}^k \left[ -\Delta_j , \gamma^{(k)}_{N,t} \right] + \frac{1}{N} \sum_{i<j}^k \left[ N^{2\beta} V (N^{\beta} (x_i -x_j)) , \gamma^{(k)}_{N,t} \right] \\ &+ \frac{N-k}{N} \sum_{j=1}^k \tr_{k+1} \left[ N^{2\beta} V (N^{\beta} (x_j - x_{k+1})), \gamma^{(k+1)}_{N,t} \right] \,
\end{split}
\end{equation}
where $\tr_{k+1}$ denotes the partial trace over the $(k+1)$-th particle. If we fix $k \geq 1$, and we let $N \to\infty$, we obtain, formally, the infinite hierarchy of equations
\begin{equation}\label{eq:infhie0}
i\partial_t \gamma^{(k)}_{\infty,t} = \sum_{j=1}^k \left[ -\Delta_j , \gamma^{(k)}_{\infty,t} \right] + b_0 \sum_{j=1}^k \tr_{k+1} \left[ \delta (x_j - x_{k+1}), \gamma^{(k+1)}_{\infty,t} \right] \,.
\end{equation}
Here we used the fact that, at least formally, in the limit $N \to \infty$ (for fixed $k\geq 1$) the second term on the r.h.s. of (\ref{eq:BBGKY1}) vanishes (because of the prefactor $1/N$), and, in the third term on the r.h.s. of (\ref{eq:BBGKY1}), $(N-k)/N \to 1$ and $N^{2\beta} V (N^{\beta} (x_i -x_j)) \to b_0 \delta (x_i -x_j)$ (with $b_0 = \int_{\Lambda} V(x) \rd x$). The infinite hierarchy (\ref{eq:infhie0}) can be written in integral form as
\begin{equation}\label{eq:infhie1}
\gamma^{(k)}_{\infty,t} = \cU^{(k)} (t) \gamma^{(k)}_{\infty,0} -i b_0 \sum_{j=1}^k \int_0^t \rd s \, \cU^{(k)} (t-s) B_{j,k+1} \gamma_{\infty,s}^{(k+1)}
\end{equation}
where we defined \begin{equation}\label{eq:free} \cU^{(k)} (t) \gamma^{(k)} = e^{it\sum_{j=1}^k \Delta_j} \gamma^{(k)} e^{-it\sum_{j=1}^k \Delta_j} \end{equation} and $B_{j,k+1}$ denotes the collision operator
\begin{equation}
B_{j,k+1} \gamma^{(k+1)} = \tr_{k+1} \, \left[ \delta (x_j -x_{k+1}) , \gamma^{(k+1)} \right] \,.
\end{equation}
In terms of kernels, $B_{j,k+1}$ (which maps $(k+1)$-particle operators into $k$-particle operators) acts as follows:
\begin{equation}\label{eq:Bk}
\left( B_{j,k+1} \gamma^{(k+1)} \right) (\bx_k; \bx'_k) = \int \rd x_{k+1} \, \left( \delta (x_j -x_{k+1}) - \delta (x'_j -x_{k+1}) \right) \gamma^{(k+1)} (\bx_k,x_{k+1}; \bx'_k , x_{k+1}) \,. \end{equation}
It is simple to check that the factorized densities $\gamma^{(k)}_{\infty,t} = |\ph_t \rangle \langle \ph_t|^{\otimes k}$, for $k \geq 1$, are a solution to the infinite hierarchy (\ref{eq:infhie1}) if $\ph_t$ solves the nonlinear Schr\"odinger equation (\ref{eq:nls}). Therefore, to prove Theorem \ref{thm:main}, we need to identify the limit of $\Gamma_{N,t} = \{ \gamma^{(k)}_{N,t} \}_{k=1}^N$ as the unique solution to (\ref{eq:infhie1}); if we can prove, namely, that every limit point of $\Gamma_{N,t}$ (with respect to an appropriate weak topology) solves (\ref{eq:infhie1}) and that the solution to (\ref{eq:infhie1}) is unique, the claim  (\ref{eq:claim}) follows by a compactness argument.

\medskip

This strategy was introduced in this context by Spohn, who applied it in \cite{Sp} to obtain a rigorous derivation of the nonlinear Hartree equation
\[ i\partial_t \ph_t = -\Delta \ph_t + (V*|\ph_t|^2)\ph_t \] for the time evolution of an initially factorized $N$-particle wave function $\psi_N = \ph^{\otimes N}$ with respect to the mean-field Hamiltonian \[ H_N^{\text{mf}} = \sum_{j=1}^N -\Delta_j + \frac{1}{N} \sum_{i<j}^N V (x_i -x_j) \] with bounded potential $V \in L^{\infty} (\bR^{d})$. In \cite{EY}, Erd\"os and Yau extended the result of Spohn to mean-field models with Coulomb interaction $V(x) = \pm 1/ |x|$ (partial results for the Coulomb potential have also been obtained in \cite{BGM}). In \cite{ES}, this strategy was applied to the study of systems of gravitating particles with relativistic dispersions (modeling boson stars).     
                                                                                         
\medskip

More recently, models with $N$-dependent potentials $V_N$ approaching a delta-function as $N \to\infty$ have been studied. Consider the time evolution $\psi_{N,t} = e^{-iH_{N,\beta} t} \psi_N$ of a factorized initial data $\psi_N = \ph^{\otimes N} \in L^2 (\bR^{3N})$ (only asymptotic factorization in the sense of (\ref{eq:asyfact}) is actually needed) with respect to the Hamiltonian \[ H_{N,\beta} = \sum_{j=1}^N -\Delta_j + \frac{1}{N} \sum_{i<j}^N N^{3\beta} V (N^{\beta} (x_i -x_j)) \] for $0< \beta \leq 1$. It follows from \cite{EESY,ESY2,ESY,ESY3} that, if $\gamma^{(k)}_{N,t}$ denotes the $k$-particle marginal associated with $\psi_{N,t}$, one has, for every $t \in \bR$ and $k \in \bN$, \[ \gamma^{(k)}_{N,t} \to |\ph_t \rangle \langle \ph_t|^{\otimes k} \] as $N \to \infty$, where $\ph_t$ solves the nonlinear Schr\"odinger equation \begin{equation}\label{eq:nls3d} i\partial_t \ph_t = -\Delta \ph_t + \sigma |\ph_t|^2 \ph_t \end{equation} with coupling constant $\sigma = \int V (x) \rd x$ if $0<\beta <1$, and $\sigma = 8 \pi a_0$ if $\beta =1$. Here $a_0$ denotes the scattering length of the potential $V$; the emergence of the scattering length in the case $\beta =1$ is a consequence of the singular correlation structure developed by the solution to the $N$-particle Schr\"odinger equation. In \cite{ABGT,AGT}, one-dimensional models with Hamiltonian \begin{equation}\label{eq:1d} H_{N,\beta} = \sum_{j=1} -\Delta_j + \frac{1}{N} \sum_{i<j} N^{\beta} V (N^{\beta} (x_i -x_j)) \end{equation} acting on $L^2_s (\bR^N)$ have been considered; for such models it was shown that the time evolution of factorized initial data can be described in terms of the one-dimensional cubic nonlinear Schr\"odinger equation with coupling constant in front of the nonlinearity given by $\int V(x) \rd x$; in this case, the correlations developed by the solution of the $N$-particle Schr\"odinger equation do not play an important role.

\medskip

Also in the-two dimensional problem discussed in the present paper, the correlations among the particles do not affect the macroscopic dynamics of the system (this explains why the coupling constant in front of the nonlinearity in (\ref{eq:nls}) is just the integral of the potential). On the contrary, the correlation structure would be very important in the study of two-dimensional systems in the Gross-Pitaevskii scaling limit (where the scattering length of the interaction potential is exponentially small in the number of particles). In \cite{LSY}, Lieb, Seiringer, and Yngvason proved that, in this limit, the ground state energy per particle can be obtained by the minimization of the so called Gross-Pitaevskii energy functional. In \cite{LS}, it was then shown by Lieb and Seiringer that the ground state vector, in the Gross-Pitaevskii limit, exhibits complete Bose Einstein condensation.
In order to prove these two results, it was very important to identify the short scale correlation structure in the ground state wave function (the energy of factorized wave functions, with absolutely no correlations, is too large by a factor of $N$). 
Unfortunately, we are not yet able to study the dynamics of Bose-Einstein condensates in the two-dimensional Gross-Pitaevskii scaling limit; nevertheless, since the infinite hierarchy which is expected to describe the time-evolution of the limiting densities $\{ \gamma^{(k)}_{\infty,t}\}_{k \geq 1}$ is still given by (\ref{eq:infhier}) (with a different coupling constant), our proof of the uniqueness of the solution of the infinite hierarchy (see Theorem \ref{thm:unique-R2} and Theorem \ref{thm:unique-box}) can also be applied to the Gross-Pitaevskii scaling limit (but, of course, in order to use Theorems \ref{thm:unique-R2} and  \ref{thm:unique-box} to prove a statement similar to (\ref{eq:claim}) in the Gross-Pitaevkii scaling, one would need to show strong bounds like (\ref{eq:aprik-inf}) on the limiting densities $\{ \gamma^{(k)}_{\infty,t}\}_{k \geq 1}$).

\medskip

As  mentioned above, the main novelty of the present paper is that we can handle systems defined on a square with periodic boundary conditions. The major difficulty in extending the derivation of the cubic nonlinear Schr\"odinger
equations (\ref{eq:nls3d}) to systems defined on a periodic domain is proving the uniqueness of the infinite hierarchy. The proof of the uniqueness given in \cite{ESY2} is based on a diagrammatic expansion of the solution of the infinite hierarchy in terms of Feynman graphs; the value of every Feynman graph was then expressed in terms of a Fourier integral, and the main part of the analysis was devoted to the control of these integrals.
For systems defined on a periodic domain, these integrals would be replaced by sums, and the analysis would be more involved; it is not yet clear if, in the case of systems defined on finite volumes, this approach can be used to prove the uniqueness of the infinite hierarchy. 

\medskip

Here we follow a different approach, first proposed in \cite{KM} by Klainerman and Machedon for three-dimensional systems on $\mathbb{R}^3$. This approach still employs the expansion introduced in \cite{ESY2}, but it then makes use of a space-time estimate for the free Schr\"odinger evolution of the densities, a simpler approach than the analysis of the contributions to the expansion in \cite{ESY2}. Moreover this approach is more suitable for certain multilinear periodic estimates first introduced by Bourgain in \cite{B1} for the study of well-posedness for the periodic Schr\"odinger equations. 

\medskip

With their argument Klainerman and Machedon obtain uniqueness of the infinite hierarchy in a class of densities satisfying certain space-time estimates. Unfortunately, in the three-dimensional setting, it is not clear if limit points $\Gamma_{\infty,t} = \{ \gamma^{(k)}_{\infty,t}\}_{k\geq 1}$ of the sequence of marginal densities $\Gamma_{N,t} = \{ \gamma^{(k)}_{N,t} \}_{k=1}^N$ satisfy these space-time estimates; for this reason, the method of Klainerman and Machedon cannot be used, in three dimensions, for deriving the nonlinear Schr\"odinger equation (\ref{eq:nls3d}). In two dimensions, however, we are able to prove (see Theorem \ref{thm:KM-bd}) that the limiting densities do indeed satisfy the space-time estimates needed as input for the analysis of Klainerman and Machedon. 

\medskip

Thus, to conclude the proof of (\ref{eq:claim}), we only have to extend their analysis to two-dimensional systems. This is not so difficult in the case $\Lambda = \bR^2$ (see Section \ref{sec:conti}), but it requires more care in the case $\Lambda = [-L,L]^{\times2}$ (see Section \ref{sec:per}). In fact, as mentioned above, for systems defined on a square, we need to use techniques from analytic number theory as in the work of Bourgain \cite{B1}; a similar approach was used also  in \cite{DPST}. The necessary number theory techniques come from \cite{BP}; see also  \cite[$\S$ 23.1]{H} and \cite{IR}.

\medskip

We should immediately remark that if we consider irrational tori, or other general non-square boxes, the argument we present here is not enough to obtain uniqueness. This is because the number of lattice points on a sphere  is precisely approximated by the Gauss lemma, but the number of lattice points on ellipsoids has no such precise approximation (see \cite{BP} and  \cite{H}). We also remark that if one considers a three-dimensional box $\Lambda = [-L,L]^{\times3}$, then the fundamental estimates we prove in Section \ref{sec:per} seem not to be available since too much regularity is lost. Such a loss of regularity  is in line with a  conjecture made by Bourgain in \cite{B1} about certain periodic Strichartz-type estimates. 

\medskip

Finally we note that with our arguments  we could prove Theorem  \ref{thm:main} also  in the one-dimensional case when  $ \Lambda = \bR$, a result already obtained by \cite{BGM}, and  when  $\Lambda = [-L,L]$.

\section{Proof of Theorem \ref{thm:main}}\label{sec:proof}
\setcounter{equation}{0}

The strategy for the proof of Theorem \ref{thm:main} is the same as the one used in \cite{ESY}. For completeness, we repeat here the main steps.

\medskip

We start by introducing an appropriate topology on density matrices. Let $\cK_k \equiv \cK (L^2 (\Lambda^k))$ be the space of compact operators on $L^2 (\Lambda^k)$, equipped with the operator norm topology, and let $\cL^1_k \equiv \cL^1 (L^2 (\Lambda^k))$ be the space of trace class operators on $L^2 (\Lambda^{k})$ equipped with the trace norm. It is well known that $\cL^1_k = \cK_k^*$. Since $\cK_k$ is separable, there exists a sequence $\{J^{(k)}_i\}_{i \ge 1} \in \cK_k$, with $\|J^{(k)}_i \| \leq 1$ for all $i \ge 1$, dense in the unit ball of $\cK_k$. On $\cL^1_k
\equiv \cL^1 (L^2 (\Lambda^k))$, we define the metric $\eta_k$ by
\begin{equation}\label{eq:etak}
\eta_k (\gamma^{(k)}, \bar \gamma^{(k)}) : = \sum_{i=1}^\infty
2^{-i} \left| \tr \; J^{(k)}_i \left( \gamma^{(k)} - \bar
\gamma^{(k)} \right) \right| \, .
\end{equation}
The topology induced by the metric $\eta_k$ and
the weak* topology are equivalent on the unit ball of $\cL^1_k$
(see \cite{Ru}, Theorem 3.16) and hence on any ball of finite radius as well.
In other words, a uniformly bounded sequence $\gamma_N^{(k)} \in \cL^1_k$ converges to
$\gamma^{(k)} \in \cL^1_k$ with respect to the weak* topology, if
and only if $\eta_k (\gamma^{(k)}_N , \gamma^{(k)}) \to 0$ as $N \to
\infty$.

\medskip

For fixed $T > 0$, let $C ([0,T], \cL^1_k)$ be the space of
functions of $t \in [0,T]$ with values in $\cL^1_k$ which are
continuous with respect to the metric $\eta_k$. On $C ([0,T],
\cL^1_k)$ we define the metric
\begin{equation}\label{eq:whetak}
\widehat \eta_k (\gamma^{(k)} (\cdot ) , \bar \gamma^{(k)} (\cdot ))
:= \sup_{t \in [0,T]} \eta_k (\gamma^{(k)} (t) , \bar \gamma^{(k)}
(t))\,.
\end{equation}
Finally, we denote by $\tau_{\text{prod}}$ the topology on the
space $\bigoplus_{k \geq 1} C([0,T], \cL^1_k)$ given by the product
of the topologies generated by the metrics $\wh \eta_k$ on $C([0,T],
\cL^1_k)$.

\begin{proof}[Proof of Theorem \ref{thm:main}]
The proof is divided in five steps.

\medskip

{\it Step 1. Approximation of the initial wave function.} Since the a-priori bounds that we are going to use are based on energy estimates, we need the expectation of $H_N^k$ at time $t=0$ to be of the order $N^k$, for all $k \geq 1$. To this end, we approximate the initial $N$-particle wave functions, by cutting off its high energy part. Let \begin{equation}\label{eq:wtpsi} \wt \psi_N = \frac{\chi ( \kappa H_N / N) \psi_N}{\| \chi (\kappa H_N / N) \psi_N \|} \, . \end{equation}
Then, for all $\kappa >0$ small enough, there exists a constant $C >0$ (of course, depending on $\kappa$) such that
\begin{equation}\label{eq:enk0}
\langle \wt \psi_N , H_N^k \wt \psi_N \rangle \leq C^k \, N^k
\end{equation}
for all $N \in \bN$, and $k \geq 1$. Moreover, using the assumption (\ref{eq:en0}), it is simple to check that
\begin{equation}\label{eq:kappa}
\| \psi_{N,t} - \wt \psi_{N,t} \| = \| \psi_N - \wt \psi_N \| \leq C \kappa
\end{equation}
uniformly in $N$. Finally, one can prove that, for every fixed $\kappa >0$ small enough,
\begin{equation}\label{eq:wtgamma}
\wt \gamma_N^{(k)} \to |\ph \rangle \langle \ph|^{\otimes k} \qquad \text{as } N \to \infty \end{equation} in the trace norm-topology. The proof of (\ref{eq:enk0}), (\ref{eq:kappa}), and (\ref{eq:wtgamma}) can be found in \cite[Proposition 9.1]{ESY3}. Next, in Steps 2-4, we will prove the statement of Theorem \ref{thm:main} for the initial data $\wt \psi_N$ with an arbitrary but fixed $\kappa >0$. Then, using (\ref{eq:kappa}), we will show in Step 5 how to obtain a proof of (\ref{eq:claim}) letting $\kappa \to 0$.

\bigskip

{\it Step 2. Compactness.} We fix $T >0$ and work in the time-interval $t \in [0,T]$ (negative times can be handled similarly). In Theorem \ref{thm:compact} we prove that, for any fixed $\kappa >0$, the sequence $\wt \Gamma_{N,t} = \{ \wt \gamma^{(k)}_{N,t} \}_{k=1}^N  \in \bigoplus_{k \geq 1} C([0,T], \cL_k^1)$ is compact with respect to the product topology $\tau_{\text{prod}}$ generated by the metrics $\wh \eta_k$. Moreover, we prove that for an arbitrary limit point
$ \Gamma_{\infty,t} = \{ \gamma_{\infty,t}^{(k)} \}_{k \geq 1}$ of the sequence $\wt \Gamma_{N,t}$, $ \gamma^{(k)}_{\infty,t}$ is symmetric w.r.t. permutations, $ \gamma^{(k)}_{\infty,t} \geq 0$, and $\tr \; \gamma^{(k)}_{\infty,t} \leq 1$
for every $k \geq 1$.

\medskip

In Theorem \ref{thm:KM-bd} we also prove that an arbitrary limit point $\Gamma_{\infty,t} = \{ \gamma_{\infty,t}^{(k)} \}_{k \geq 1}$ satisfies the a-priori estimates
\begin{equation}\label{eq:aprik}
\left\| S^{(k,\alpha)} B_{j,k+1} \gamma^{(k+1)}_{\infty,t} \right\|_{L^2 (\Lambda^k \times \Lambda^k)} \leq C^k
\end{equation}
for all $k \geq 1$, $j=1, \dots,k$, and for all $t \in [0,T]$. Here we used the notation \begin{equation}\label{eq:Ska}
S^{(k,\alpha)}= \prod_{j=1}^k (1-\Delta_{x_j})^{\alpha/2} (1-\Delta_{x'_j})^{\alpha/2} \end{equation}
and the collision operator $B_{j,k+1}$ is defined in (\ref{eq:Bk}). Note that, with a slight abuse of notation, we identify, in (\ref{eq:aprik}), the operator $S^{(k,\alpha)} B_{j,k+1} \gamma^{(k+1)}_{\infty,t}$ with its kernel $\left( S^{(k,\alpha)} B_{j,k+1} \gamma^{(k+1)}_{\infty,t} \right) (\bx_k ; \bx'_k)$.

\bigskip

{\it Step 3. Convergence.} In Theorem \ref{thm:convergence}, we show that an arbitrary limit point $\Gamma_{\infty,t} = \{ \gamma^{(k)}_{\infty,t} \}_{k \geq 1} \in \bigoplus_{k \geq 1} C([0,T], \cL^1_k )$ of the sequence $\wt \Gamma_{N,t}$ is a solution to the infinite hierarchy of equations
\begin{equation}\label{eq:infhier}
\gamma^{(k)}_{\infty,t} = \cU^{(k)} (t) \gamma^{(k)}_{\infty,0} -i b_0 \sum_{j=1}^k \int_0^t \rd s \, \cU^{(k)} (t-s) B_{j,k+1} \gamma^{(k+1)}_{\infty,s}
\end{equation}
with initial data $\gamma^{(k)}_{\infty,0} = |\ph \rangle \langle \ph|^{\otimes k}$. Here the free evolution $\cU^{(k)} (t)$ is defined in (\ref{eq:free}) and the map $B_{j,k+1}$ in (\ref{eq:Bk}).

\medskip

Note that the infinite hierarchy (\ref{eq:infhier}) has a factorized solution. In fact, it is simple to check that the family $\{ \gamma^{(k)}_t \}_{k \geq 1}$ with $\gamma^{(k)}_t = |\ph_t \rangle \langle \ph_t|^{\otimes k}$ for all $k \geq 1$ is a solution to (\ref{eq:infhier}) (with the correct initial data) provided that $\ph_t$ solves the nonlinear Schr\"odinger equation \[ i\partial \ph_t = -\Delta \ph_t + b_0 |\ph_t|^2 \ph_t \, . \]

\bigskip

{\it Step 4. Uniqueness.} In Theorem \ref{thm:unique-R2} (for the case $\Lambda = \bR^2$) and in Theorem \ref{thm:unique-box} (for the case $\Lambda = [-L, L]^{\times 2}$), we prove that the solution to the infinite hierarchy (\ref{eq:infhier}) is unique in the space of densities satisfying the a-priori estimates (\ref{eq:aprik}). More precisely, we prove that, given a family $\Gamma = \{ \gamma^{(k)} \}_{k \geq 1} \in \bigoplus_{k \geq 1} \cL^1_k$, there exists at most one solution $\Gamma_t = \{ \gamma^{(k)}_t \}_{k\geq 1} \in \bigoplus_{k \geq 1} C([0,T], \cL^1_k)$ of (\ref{eq:infhier}) such that
\[  \left\| S^{(k,\alpha)} B_{j,k+1} \gamma^{(k+1)}_{\infty,t} \right\|_{L^2 (\Lambda^k \times \Lambda^k)} \leq C^k \]  for all $k \geq 1$ and $t \in [0,T]$.

\medskip

{\it Step 5. Conclusion of the proof.} Combining the results of Step 2--Step 4, we immediately obtain that, for every fixed $\kappa >0$, $\wh\eta (\wt \gamma^{(k)}_{N,t} , |\ph_t \rangle \langle \ph_t |^{\otimes k}) \to 0$ as $N \to\infty$ for every fixed $k \geq 1$. In particular this implies that, for every fixed $\kappa >0$, $t \in [0,T]$ and $k \geq 1$, \begin{equation}\label{eq:wtconv} \wt \gamma^{(k)}_{N,t} \to |\ph_t \rangle \langle \ph_t |^{\otimes k}\end{equation} with respect to the weak* topology of $\cL^1_k$. To prove that also $\gamma_{N,t}^{(k)}$, that is the $k$-particle marginal density associated with the original initial wave functions $\psi_N$, converges to the projection $|\ph_t \rangle \langle \ph_t|^{\otimes k}$ as $N \to \infty$, we observe that, for any fixed $\eps >0$ and for every compact operator $J^{(k)} \in \cK_k$, we can find, by (\ref{eq:kappa}), a sufficiently small $\kappa >0$ such that
\[ \left| \tr \;  J^{(k)} \left( \gamma_{N,t}^{(k)} - \wt \gamma_{N,t}^{(k)} \right) \right| \leq \| J^{(k)} \| \, \| \psi_N - \wt \psi_N \| \leq C \kappa \leq \eps/2 \, ,\] uniformly in $N \in \bN$. For this fixed value of $\kappa >0$ we obtain, from (\ref{eq:wtconv}), that
\[ \left| \tr \; J^{(k)} \left( \wt \gamma_{N,t}^{(k)} - |\ph_t \rangle \langle \ph_t|^{\otimes k} \right) \right| \leq  \eps /2 \] for all $N$ large enough. This proves that, for arbitrary $\eps >0$ and $J^{(k)} \in \cK_k$ there exists $N_0 >0$ such that
\[ \left| \tr \;  J^{(k)} \left( \gamma_{N,t}^{(k)} -  |\ph_t \rangle \langle \ph_t|^{\otimes k}  \right) \right| \leq \eps \] for all $N >N_0$. This proves that, for every fixed $t \in [0,T]$, and $k \geq 1$, $\gamma^{(k)}_{N,t} \to |\ph_t \rangle \langle \ph_t|^{\otimes k}$ as $N \to \infty$, with respect to the weak* topology of $\cL^1_k$. Since, however, the limiting density is an orthogonal projection, the convergence in the weak* topology is equivalent to the convergence in the trace norm topology. This concludes the proof of Theorem \ref{thm:main}.
\end{proof}

\section{Energy Estimates and A-Priori Bounds on $\wt\Gamma_{N,t} = \{ \wt\gamma^{(k)}_{N,t} \}_{k=1}^N$}
\setcounter{equation}{0}

\begin{proposition}\label{prop:enest}
Suppose that the Hamiltonian $H_N$ is defined as in (\ref{eq:ham}), with $0< \beta < 3/4$. Then there exists a constant $C >0$, and, for every $k \geq 0$, there exists $N_0 = N_0 (k)$ such that \begin{equation}\label{eq:Hkest} \langle \psi, (H_N +N)^k \psi \rangle \geq C^k N^k \langle \psi, (1-\Delta_{x_1}) \dots (1-\Delta_{x_k}) \, \psi \rangle \end{equation} for all $N \geq N_0$, and all $\psi \in L^2_s (\Lambda^{N})$.
\end{proposition}
\begin{proof}
We proceed by a two-step induction over $k \geq 0$. For $k=0$ the statement is trivial and for $k=1$ it follows from the positivity of the potential. Suppose the claim holds for all $k \leq n$. We prove it holds for $k=n+2$. In fact, from the induction assumption, and using the notation $S_i = (1- \Delta_{x_i})^{1/2}$, we find
\begin{equation}
\langle \psi, (H_N +N)^{n+2} \psi \rangle \geq C^n N^n \langle \psi, (H_N +N) S_1^2 \dots S_n^2 (H_N + N) \psi \rangle\,.
\end{equation}
Now, writing $H_N+N = h_1 + h_2$, with \[ h_1 = \sum_{j=k+1}^N S_j^2 \qquad \text{and } \qquad h_2 = \sum_{j=1}^k S_j^2 + \sum_{i<j}^N N^{2\beta-1} V(N^{\beta} (x_i - x_j)) \]
it follows that
\begin{equation}
\begin{split}
\langle \psi, &(H_N +N)^{n+2} \psi \rangle \\ \geq \; &C^n N^n \langle \psi, h_1 S_1^2 \dots S_n^2 h_1 \psi \rangle \\ &+ C^n N^n \left( \langle \psi, h_1 S_1^2 \dots S_n^2 h_2 \psi \rangle + \langle \psi, h_2 S_1^2 \dots S_n^2 h_1 \psi \rangle \right)   \\
\geq \; &C^n N^n (N-n) (N-n-1) \langle \psi, S_1^2 \dots S_{n+2}^2 \psi \rangle + C^n N^n (N-n) \langle \psi, S_1^4 S_2^2 \dots S_{n+1}^2 \psi \rangle \\
&+ C^n N^n \frac{(N-n)}{N} N^{2\beta} \sum_{i<j}^N \left( \langle \psi, S_1^2 \dots S_{n+1}^2 V (N^{\beta} (x_i -x_j)) \psi \rangle + \text{complex conjugate} \right)
\end{split}
\end{equation}
Because of the permutation symmetry of $\psi$, we obtain
\begin{equation}\label{eq:bd}
\begin{split}
\langle \psi, &(H_N +N)^{n+2} \psi \rangle \\
\geq \; &C^{n+2} N^{n+2} \langle \psi, S_1^2 \dots S_{n+2}^2 \psi \rangle + C^{n+1} N^{n+1}  \langle \psi, S_1^4 S_2^2 \dots S_{n+1}^2 \psi \rangle \\
&+ C^n N^{n-1} N^{2\beta} (N-n)^2 (N-n-1) \left( \langle \psi, S_1^2 \dots S_{n+1}^2 V (N^{\beta} (x_{n+2} -x_{n+3})) \psi \rangle + \text{c.c.} \right) \\
&+ C^n N^{n-1} N^{2\beta} (N-n)^2 (n+1) \left(  \langle \psi, S_1^2 \dots S_{n+1}^2 V (N^{\beta} (x_{1} -x_{n+2})) \psi \rangle + \text{c.c.} \right) \\
&+ C^n N^{n-1} N^{2\beta} (N-n)(n+1)n \left(  \langle \psi, S_1^2 \dots S_{n+1}^2 V (N^{\beta} (x_{1} -x_{2})) \psi \rangle + \text{c.c.} \right)
\end{split}
\end{equation}
The last three terms are the errors we need to control. First of all, we remark that the first error term is positive, and thus can be neglected (because we assumed $V \geq 0$). In fact, since $V(N^{\beta} (x_{n+2} - x_{n+3}))$ commutes with all derivatives $S_1,\dots ,S_n$, we have
\[ \langle \psi, S_1^2 \dots S_{n+1}^2 V (N^{\beta} (x_{n+2} -x_{n+3})) \psi \rangle = \int \rd \bx \; V(N^{\beta} (x_{n+2} - x_{n+3}) |(S_1 \dots S_{n+1} \psi) (\bx)|^2 \geq 0 \, . \]
As for the second error term on the r.h.s. of (\ref{eq:bd}), we bound it from below by
\begin{equation}
\begin{split}
C^n N^{n-1} N^{2\beta} &(N-n)^2 (n+1) \left(  \langle \psi, S_1^2 \dots S_{n+1}^2 V (N^{\beta} (x_{1} -x_{n+2})) \psi \rangle + \text{c.c.} \right) \\ \geq \; &- C(n) N^{n+1} N^{2\beta} \left| \langle \psi, S_{n+1} \dots S_2 S_1 \, [ S_1, V(N^{\beta} (x_1 -x_{n+2}) ] \, S_2 \dots S_{n+1} \psi \rangle \right|  \\  \geq \; &- C(n) N^{n+1} N^{3\beta} \left| \langle \psi, S_{n+1} \dots S_2 S_1 \, (\nabla V)(N^{\beta} (x_1 -x_{n+2}) \, S_2 \dots S_{n+1} \psi \rangle \right| \\ \geq \; &- C (n) N^{n+1} N^{3\beta} \left( \langle \psi, S_{n+1} \dots S_2 S_1 \, \left| (\nabla V)(N^{\beta} (x_1 -x_{n+2})) \right| S_1 S_2 \dots S_{n+1} \psi \rangle \right. \\ &\hspace{4cm} \left. + \langle  \psi, S_{n+1} \dots S_2 \, \left| (\nabla V)(N^{\beta} (x_1 -x_{n+2})) \right| S_2 \dots S_{n+1} \psi \rangle \right)
\end{split}
\end{equation}
for a constant $C(n)$ independent of $N$. Using that
\[ \langle \psi, V(x) \psi \rangle \leq \; C \| V \|_{p} \, \langle \psi, (1-\Delta) \psi \rangle \] for every $p >1$ (see Lemma \ref{lm:sob}) we find
\begin{equation}\label{eq:err0}
\begin{split}
C^n N^{n-1} N^{2\beta} &(N-n)^2 (n+1) \left(  \langle \psi, S_1^2 \dots S_{n+1}^2 V (N^{\beta} (x_{1} -x_{n+2})) \psi \rangle + \text{c.c.} \right)
\\ \geq \; &-C (n) N^{n+1} N^{3 \beta} N^{-2\beta + \eps} \langle \psi, S_1^2 \dots S_{n+2}^2 \psi \rangle = - C(n) N^{n+1+\beta+\eps} \langle \psi, S_1^2 \dots S_{n+2}^2 \psi \rangle
\end{split}
\end{equation}
for arbitrary $\eps >0$. The last term on the r.h.s. of (\ref{eq:bd}), on the other hand, can be controlled by
\begin{equation}\label{eq:bd2}
\begin{split}
C^n N^{n-1} N^{2\beta} &(N-n)(n+1)n \left(  \langle \psi, S_1^2 \dots S_{n+1}^2 V (N^{\beta} (x_{1} -x_{2})) \psi \rangle + \text{c.c.} \right) \\
\geq \; &- C(n) N^{n} N^{2\beta} \left| \langle \psi, S_{n+1} \dots S_2 S_1 \, [ S_1 S_2, V(N^{\beta} (x_1 -x_{n+2})) ] \, S_2 \dots S_{n+1} \psi \rangle \right|  \\  \geq \; &- C(n) N^{n} N^{3\beta} \left| \langle \psi, S_{n+1} \dots S_2 S^2_1 \, (\nabla V)(N^{\beta} (x_1 -x_{2})) \, S_3 \dots S_{n+1} \psi \rangle \right|  \\ &-C(n) N^n N^{3\beta} \left| \langle \psi, S_{n+1} \dots S_2 S_1 \, (\nabla V)(N^{\beta} (x_1 -x_{2})) \, S_2 \dots S_{n+1} \psi \rangle \right|
\end{split}
\end{equation}
The second term is bounded by
\begin{equation}\label{eq:err1}
\begin{split}
-C(n) &N^n N^{3\beta} \left| \langle \psi, S_{n+1} \dots S_2 S_1 \, (\nabla V)(N^{\beta} (x_1 -x_{2})) \, S_2 \dots S_{n+1} \psi \rangle \right| \\ \geq \; &-C (n) N^n N^{3\beta} \left(  \alpha \, \langle \psi, S_{n+1} \dots S_1 \, \left| (\nabla V)(N^{\beta} (x_1 -x_{2})) \right| \, S_1 \dots S_{n+1} \psi \rangle \right. \\ &\hspace{3cm} \left. + \alpha^{-1} \langle \psi, S_{n+1} \dots S_2 \, \left| (\nabla V)(N^{\beta} (x_1 -x_{2})) \right| \, S_2 \dots S_{n+1} \psi \rangle \right) \\ \geq \; &-C(n) N^n N^{3\beta} \left(\alpha \langle \psi, S^2_1 \dots S^2_{n+1} \psi \rangle + \alpha^{-1} N^{-2\beta+\eps} \langle \psi, S_1^2 \dots S_{n+1}^2 \psi \rangle \right) \\ \geq \; &-C(n) N^n N^{2\beta + \eps} \langle \psi, S_1^2 \dots S_{n+1}^2 \psi \rangle
\end{split}
\end{equation}
where, in the last inequality we optimized the choice of $\a$, by putting $\a = N^{-\beta}$.
The first term on the r.h.s. of (\ref{eq:bd2}), on the other hand, is controlled by
\begin{equation}\label{eq:err2}
\begin{split}
- C(n) &N^{n} N^{3\beta} \left| \langle \psi, S_{n+1} \dots S_2 S^2_1 \, (\nabla V)(N^{\beta} (x_1 -x_{2})) \, S_3 \dots S_{n+1} \psi \rangle \right|
\\ \geq \; & - C (n) N^{n} N^{3\beta} \left( \alpha \, \langle \psi, S_{n+1} \dots S_2 S_1^2 \left|  (\nabla V)(N^{\beta} (x_1 -x_{2})) \right| S_1^2 S_2 \dots S_{n+1} \psi \rangle \right. \\ &\hspace{3cm} \left.+ \alpha^{-1} \langle \psi, S_{n+1} \dots S_3 \left|  (\nabla V)(N^{\beta} (x_1 -x_{2})) \right| S_3 \dots S_{n+1} \psi \rangle \right) \\ \geq \; &-C (n) N^n N^{3\beta} \left( \alpha \langle \psi, S_1^4 S^2_2 \dots S^2_{n+1} \psi \rangle + \alpha^{-1} N^{-2\beta} \langle \psi, S_1^2 \dots S_{n+1}^2 \psi \rangle \right) \\ \geq \; &- C(n) N^{n+4\beta-2+\eps} \langle \psi, S_1^4 S^2_2 \dots S^2_{n+1} \psi \rangle - C(n) N^{n+2-\eps} \langle \psi, S_1^2 \dots S^2_{n+1} \psi \rangle
\end{split}
\end{equation}
where we chose $\alpha = N^{-2 + \beta +\eps}$, for some $\eps >0$. Inserting (\ref{eq:err1}) and (\ref{eq:err2}) on the r.h.s. of (\ref{eq:bd2}), we find
\begin{equation}\label{eq:err4}
\begin{split}
C^n N^{n-1} N^{2\beta} &(N-n)(n+1)n \left(  \langle \psi, S_1^2 \dots S_{n+1}^2 V (N^{\beta} (x_{1} -x_{2})) \psi \rangle + \text{c.c.} \right) \\
\geq \; &-C(n) N^n N^{2\beta + \eps} \langle \psi, S_1^2 \dots S_{n+1}^2 \psi \rangle
\\ &- C(n) N^{n+4\beta-2+\eps} \langle \psi, S_1^4 S^2_2 \dots S^2_{n+1} \psi \rangle - C(n) N^{n+2-\eps} \langle \psi, S_1^2 \dots S^2_{n+1} \psi \rangle
\end{split}
\end{equation}
Inserting (\ref{eq:err0}) and (\ref{eq:err4}) on the r.h.s. of (\ref{eq:bd}), we see that, for $\beta < 3/4$ (choosing $\eps >0$ small enough) all error terms can  be controlled by the two positive contributions, and the proposition follows.
\end{proof}

{F}rom these energy estimates, we immediately obtain strong a-priori bounds on the marginal densities $\wt \gamma^{(k)}_{N,t}$.

\begin{corollary}\label{cor:apriNk}
Let $\wt \psi_{N,t} = e^{-iH_N t} \wt \psi_N$ be the solution of the $N$-particle Schr\"odinger equation with initial wave function $\wt \psi_N$, as defined in (\ref{eq:wtpsi}) (for a fixed $\kappa >0$), and let $\wt \gamma^{(k)}_{N,t}$ denote its $k$-particle marginal. Then there exists a constant $C >0$ (depending on $\kappa$) and, for every $k \geq 1$, an integer $N_0 (k)$ such that
\begin{equation}\label{eq:apriNk}
\tr \; (1-\Delta_1) \dots (1-\Delta_k) \; \wt\gamma^{(k)}_{N,t} \leq C^k
\end{equation}
for all $N > N_0 (k)$.
\end{corollary}

\begin{proof}
We have
\begin{equation}
\begin{split}
\tr\; (1-\Delta_1) \dots (1-\Delta_k) \wt\gamma^{(k)}_{N,t} = \; &\langle \wt \psi_{N,t} , S_1^2 \dots S_k^2 \; \wt\psi_{N,t} \rangle \\  \leq \; & \frac{1}{C^k N^k} \langle \wt \psi_{N,t}, H_N^k \wt \psi_{N,t} \rangle  = \frac{1}{C^k N^k} \langle \wt \psi_N, H_N^k \wt\psi_N \rangle \leq C^k
\end{split}
\end{equation}
where in the first inequality we used Proposition \ref{prop:enest}, and in the last inequality we used (\ref{eq:enk0}).
\end{proof}

\section{Compactness of the sequence $\wt \Gamma_{N,t} = \{ \wt \gamma_{N,t}^{(k)} \}_{k =1}^N$}
\setcounter{equation}{0}

\begin{theorem}\label{thm:compact}
Suppose that $\wt\psi_N$ is defined as in (\ref{eq:wtpsi}), let $\wt\psi_{N,t} = e^{-iH_Nt} \wt\psi_N$ and denote by $\wt \gamma^{(k)}_{N,t}$ the $k$-particle marginal density associated with $\wt \psi_{N,t}$. Then the sequence of marginal densities $\wt \Gamma_{N,t} = \{ \wt \gamma^{(k)}_{N,t} \}_{k=1}^N \in \bigoplus_{k \geq 1} C([0,T], \cL_k^1)$ is compact with respect to the product topology $\tau_{\text{prod}}$ generated by the metrics $\wh \eta_k$ (defined in Section \ref{sec:proof}). For any limit point $ \Gamma_{\infty,t} = \{ \gamma_{\infty,t}^{(k)}
\}_{k \geq 1}$, $ \gamma^{(k)}_{\infty,t}$ is symmetric w.r.t.
permutations, $ \gamma^{(k)}_{\infty,t} \geq 0$, and
\begin{equation}\label{eq:bou} \tr \; \gamma^{(k)}_{\infty,t} \leq 1
\,\end{equation} for every $k \geq 1$.
\end{theorem}

\begin{proof}
By a Cantor diagonal argument it is enough to prove the compactness of
$\wt\gamma_{N,t}^{(k)}$ for fixed $k \geq 1$ with respect to the metric $\wh \eta_k$. To
this end, we show the equicontinuity of $\gamma_{N,t}^{(k)}$ with respect to the metric $\eta_k$.
It is enough to prove (see Lemma 6.2 in \cite{ESY}) that, for every observable $J^{(k)}$ from a dense subset of $\cK_k$ and for every $\eps >0$, there exists a $\delta = \delta (J^{(k)}, \eps)> 0$ such that
\begin{equation}\label{eq:equi02}
\sup_{N\ge 1}\Big| \tr \;  J^{(k)} \left( \wt \gamma_{N,t}^{(k)} -
\wt \gamma_{N,s}^{(k)} \right) \Big| \leq \eps
\end{equation}
for all $t,s \in [0,T]$ with $|t -s| \leq \delta$. We are going to prove (\ref{eq:equi02}) for all $J^{(k)} \in \cK_k$ such that $\tri J^{(k)} \tri < \infty$, where we defined the norm
\begin{equation}\label{eq:tri}\begin{split}
\tri J^{(k)} \tri := \sup_{\bp'_k} \int \rd \bp_k \; \prod_{j=1}^k (1+p_j^2)^{1/2} (1+(p'_j)^2)^{1/2} \; \left( \left| \widehat{J}^{(k)} (\bp_k;\bp'_k) \right|  + \left| \widehat{J}^{(k)} (\bp'_k; \bp_k)\right| \right).
\end{split}\end{equation}
Here $\widehat{J}^{(k)} (\bp_k; \bp'_k)$ denotes the kernel of the compact operator $J^{(k)}$ in momentum space. It is simple to check that the subset of $\cK_k$ consisting of all $J^{(k)}$ with $\tri J^{(k)} \tri < \infty$ is dense.

\medskip

Fix now $\eps>0$ and $J^{(k)} \in \cK_k$ with $\tri J^{(k)} \tri < \infty$. Then, rewriting the BBGKY hierarchy (\ref{eq:BBGKY1}) in integral form and multiplying it with $J^{(k)}$ we obtain that, for any $r \leq t$,
\begin{equation}\label{eq:equi-1}
\begin{split}
\Big| \tr \, J^{(k)} \left(  \wt\gamma_{N,t}^{(k)} -
\wt\gamma_{N,r}^{(k)} \right) \Big| \leq \; &\sum_{j=1}^k \int_r^t \rd s \,
\Big| \tr \; J^{(k)} [ -\Delta_j , \wt\gamma_{N,s}^{(k)}] \Big| \\ &+ N^{2\beta-1} \sum_{i<j}^k
\int_r^t \rd s \, \Big| \tr \; J^{(k)} [ V (N^{\beta} (x_i -x_j)) , \wt\gamma^{(k)}_{N,s} ] \Big| \\ &+ N^{2\beta} \left(1-\frac{k}{N}\right) \sum_{j=1}^k \int_r^t \rd s \, \Big|
\tr \; J^{(k)}\left[ V (N^{\beta}(x_j - x_{k+1})), \wt\gamma^{(k+1)}_{N,s}
\right]\Big|.
\end{split}
\end{equation}
It is simple to prove that
\[ \Big| \tr \; J^{(k)} [ -\Delta_j , \wt\gamma_{N,s}^{(k)}] \Big| \leq 2 \tri J^{(k)} \tri \, \tr \, \wt \gamma_{N,s}^{(k)} \leq 2 \tri J^{(k)} \tri  \, .\]
To bound the last term on the r.h.s. of (\ref{eq:equi-1}), we observe that, using the notation $S_j = (1-\Delta_j)^{1/2}$, we have
\begin{equation}\label{eq:model}
\begin{split}
N^{2\beta} \Big| & \tr \; J^{(k)} \left[ V (N^{\beta}(x_j - x_{k+1})), \wt\gamma^{(k+1)}_{N,s}
\right]\Big| \\ = \; & N^{2\beta} \Big| \tr \; J^{(k)} V (N^{\beta}(x_j - x_{k+1})) \wt\gamma^{(k+1)}_{N,s} - \tr \; J^{(k)} \wt\gamma^{(k+1)}_{N,s}V (N^{\beta}(x_j - x_{k+1})) \Big| \\ \leq \; & N^{2\beta} \Big| \tr \; S^{-1}_j S^{-1}_{k+1} J^{(k)} S_j S_{k+1} S_j^{-1} S_{k+1}^{-1} V (N^{\beta}(x_j - x_{k+1})) S_{k+1}^{-1} S_j^{-1} S_j S_{k+1} \wt\gamma^{(k+1)}_{N,s} S_j S_{k+1} \\
&\hspace{.5cm}- \tr \; S_j S_{k+1} J^{(k)} S_j^{-1} S_{k+1}^{-1} S_{k+1} S_j \wt\gamma^{(k+1)}_{N,s} S_j S_{k+1} S_{k+1}^{-1} S_j^{-1} V (N^{\beta}(x_j - x_{k+1})) S_j^{-1} S_{k+1}^{-1} \Big| \\ \leq \; & N^{2\beta} \,\left( \left\| S_j S_{k+1} J^{(k)} S_j^{-1} S_{k+1}^{-1} \right\| + \left\| S^{-1}_j S^{-1}_{k+1} J^{(k)} S_j S_{k+1} \right\| \right) \\ &\hspace{.5cm}  \times \left\| S_{k+1}^{-1} S_j^{-1} V (N^{\beta}(x_j - x_{k+1})) S_j^{-1} S_{k+1}^{-1} \right\|  \; \sup_{s \in \bR} \tr \; S^2_j S^2_{k+1} \gamma_{N,s}^{(k+1)} \\
\leq  \; &C \tri J^{(k)} \tri  ,
\end{split}
\end{equation}
where, in the last inequality we used (\ref{eq:aprik}) and the fact that, by Lemma \ref{lm:sob},
\[ \left\| S_1^{-1} S_2^{-2} V (x_1 -x_2) S_2^{-1} S_1^{-1} \right\| \leq C \, \| V \|_1 \, .\]
The second term on the r.h.s. of (\ref{eq:equi-1}) can be handled similarly. This implies  (\ref{eq:equi02}). The proof of the fact that $\gamma_{\infty,t}^{(k)}$ is symmetric w.r.t. permutations, that it is non-negative and with $\tr \, \gamma_{\infty,t}^{(k)} \leq 1$ can be found in \cite[Theorem 6.1]{ESY}.
\end{proof}

\section{A-Priori Estimate on the Limit Points $\Gamma_{\infty,t} = \{ \gamma^{(k)}_{\infty,t} \}_{k \geq 1}$} \label{sec:apriKM}
\setcounter{equation}{0}

Since the a-priori estimates (\ref{eq:apriNk}) on $\wt \gamma_{N,t}^{(k)}$ hold uniformly in $N$, we can extract estimates on the limit points $\{ \gamma^{(k)}_{\infty,t}\}_{k \geq 1}$.

\begin{proposition}\label{cor:aprik-inf} Suppose that $\Gamma_{\infty,t} = \{ \gamma^{(k)}_{\infty,t} \}_{k\geq 1} \in \bigoplus_{k\geq 1} C([0,T], \cL^1_k)$ is a limit point of the sequence $\wt \Gamma_{N,t} = \{ \wt \gamma^{(k)}_{N,t} \}_{k=1}^N$ with respect to the product topology $\tau_{\text{prod}}$. Then there exists $C>0$ (depending on $\kappa$) such that
\begin{equation}\label{eq:aprik-inf} \tr (1-\Delta_1) \dots (1-\Delta_k) \gamma_{\infty,t}^{(k)} \leq C^k \end{equation}
for all $k \geq 1$.
\end{proposition}
\begin{proof}
The bound (\ref{eq:aprik-inf}) follows from the a-priori bound (\ref{eq:aprik}) by taking the limit $N \to \infty$. The details of the proof can be found in \cite{ESY2}.
\end{proof}

In order to apply the technique of Klainerman and Machedon (see \cite{KM}) to prove the uniqueness of the infinite hierarchy, we need different a-priori bounds on the limiting density. These are provided by the following proposition.
\begin{theorem}\label{thm:KM-bd}
Suppose that $\Gamma_{\infty,t} = \{ \gamma^{(k)}_{\infty,t} \}_{k\geq 1} \in \bigoplus_{k\geq 1} C([0,T], \cL^1_k)$ is a limit point of the sequence $\wt \Gamma_{N,t} = \{ \wt \gamma^{(k)}_{N,t} \}_{k=1}^N$ with respect to the product topology $\tau_{\text{prod}}$. Then, for every $\alpha <1$, there exists $C>0$ (depending on $\kappa$) such that
\begin{equation}\label{eq:KM} \left\| S^{(k,\alpha)} B_{j,k} \gamma^{(k+1)}_{\infty,t} \right\|_{L^2 (\Lambda^k \times \Lambda^k)} \leq C^k
\end{equation}
for all $k \geq 1$ and all $t \in [0,T]$. Here $S^{(k,\alpha)} = \prod_{j=1}^k (1-\Delta_{x_j})^{\alpha/2} (1-\Delta_{x'_j})^{\alpha/2}$.
\end{theorem}

\begin{proof}
By (\ref{eq:aprik-inf}), it is enough to prove that
\begin{equation}\label{eq:KM2}
\left\| S^{(k,\alpha)} B_{j,k+1} \gamma^{(k+1)}_{\infty,t} \right\|_{L^2 (\Lambda^k \times \Lambda^k)} \leq C \; \tr \; (1-\Delta_1) \dots (1-\Delta_{k+1}) \, \gamma^{(k+1)}_{\infty,t} \,.
\end{equation}
We only consider $k=1$ and $j=1$ (the argument for $k \geq 2$ is similar). Moreover, we focus on $\Lambda = \bR^2$; the case $\Lambda = [-L, L]^{\times2}$ can be handled similarly (sums are going to replace integrals over momenta). Switching to Fourier space we have
\begin{equation}
\begin{split}
\widehat{B_{1,2} \gamma_{\infty,t}^{(2)}} (p \, ; p') &= \int \rd x_1 \, \rd x'_1 \; e^{i x_1 \cdot p} \; e^{-i x'_1 \cdot p'} \int \rd x_{2} \rd x'_{2} \, \delta (x_1 - x_2) \delta (x_2 - x'_2) \gamma_{\infty,t}^{(2)} (x_1, x_2 ; x'_1, x'_2 ) \\
&= \int \rd q \, \rd \kappa  \int \rd x_1 \rd x_2 \rd x'_1 \rd x'_2 \; e^{i x_1 \cdot p} \; e^{-i x'_1 \cdot p'} e^{i q (x_2 -x_1)} \, e^{i\kappa (x_2 - x'_2)} \; \gamma_{\infty,t}^{(2)} (x_1, x_2 ; x'_1, x'_{2}) \\
&= \int \rd q \, \rd \kappa \; \widehat{\gamma}^{(2)}_{\infty,t} (p - q, q + \kappa; p', \kappa)
\end{split}
\end{equation}
Thus
\begin{equation}
\begin{split}
\widehat{\left(S^{(1,\alpha)} B_{1,2} \gamma_{\infty,t}^{(2)}\right)} (p; p') &= (1+p^2)^{\alpha/2} \, (1+(p')^2)^{\alpha/2} \, \int \rd q \rd \kappa \;  \widehat{\gamma}^{(2)}_{\infty,t} (p - q, q + \kappa; p', \kappa) \end{split}
\end{equation}
and
\begin{equation}\begin{split}
\left\| S^{(1,\alpha)} B_{1,2} \gamma^{(2)}_{\infty,t} \right\|_{L^2 (\Lambda \times \Lambda)}^2 = \; \int \rd p &\, \rd p' \, \rd q_1 \rd q_2 \rd \kappa_1 \rd \kappa_2 \, (1+p^2)^{\alpha} (1+ (p')^2)^{\alpha} \\ & \times  \widehat{\gamma}^{(2)}_{\infty,t} (p - q_1, q_1 + \kappa_1; p', \kappa_1) \; \widehat{\gamma}^{(2)}_{\infty,t} (p - q_2, q_2 + \kappa_2; p', \kappa_2)
\end{split}
\end{equation}
Using the decomposition
\begin{equation}\label{eq:gamma}
\widehat{\gamma}_{\infty,t}^{(2)} (p_1,p_2 ; p'_1, p'_2) = \sum_j \lambda_j \, \psi_j (p_1, p_2) \, \overline{\psi}_j (p'_1, p'_2) \end{equation} for an orthonormal family $\{ \psi_j\}$, we arrive at (notice that $\lambda_j \geq 0$ for all $j$, and $\sum_j \lambda_j  \leq 1$, because $\gamma^{(k+1)}$ is a non-negative trace-class operator with trace lesser or equal to one by Theorem \ref{thm:compact}):
\begin{equation}\label{eq:int}
\begin{split}
\left\| S^{(1,\alpha)} B_{1,2} \gamma_{\infty,t}^{(2)} \right\|_{L^2 (\Lambda \times \Lambda)}^2 = \; & \sum_{i,j} \lambda_i \lambda_j \;
\int \rd p \, \rd p' \, \rd q_1 \rd q_2 \rd \kappa_1 \rd \kappa_2 \, (1+p^2)^{\alpha}\, (1+(p')^2)^{\alpha} \\ &\times  \psi_j (p - q_1, q_1 + \kappa_1) \, \overline{\psi}_j  (p', \kappa_1) \; \psi_i (p - q_2, q_2 + \kappa_2) \, \overline{\psi}_i (p', \kappa_2)
\end{split}
\end{equation}
Next we use that
\[ (1 + p^2)^{\alpha/2}  \leq C \left( (1 + (p - q_1)^2)^{\alpha/2} + (1 + (q_1+ \kappa_1)^2)^{\alpha/2} + (1+ \kappa_1^2)^{\alpha/2} \right) \]
and that, analogously,
\[ (1 + p^2)^{\alpha/2}  \leq C \left( (1 + (p - q_2)^2)^{\alpha/2} + (1+ (q_2+ \kappa_2)^2)^{\alpha/2} + (1+ \kappa_2^2)^{\alpha/2} \right) \] to estimate
\begin{equation}\label{eq:bd22}
\begin{split}
(1+p^2)^{\alpha} \leq C &\left( (1 + (p - q_1)^2)^{\alpha/2} + (1+ (q_1+ \kappa_1)^2)^{\alpha/2} + (1+ \kappa_1^2)^{\alpha/2} \right) \\ &\times \left( (1 + (p - q_2)^2)^{\alpha/2} + (1+ (q_2+ \kappa_2)^2)^{\alpha/2} + (1+ \kappa_2^2)^{\alpha/2} \right).
\end{split}
\end{equation}
When we insert this bound in (\ref{eq:int}), we obtain 9 different contributions. We show, for example, how to control the first contribution (where we replace the factor $(1+p^2)^{\alpha}$ on the r.h.s. of (\ref{eq:int}) by $(1+(p- q_1)^2)^{\alpha/2} (1+(p-q_2)^2)^{\alpha/2}$). To this end we use a weighted Schwarz inequality, and we get
\begin{equation}\label{eq:int3}
\begin{split}
\int &\rd p \, \rd p' \, \rd q_1 \rd q_2 \rd \kappa_1 \rd \kappa_2 \, (1+(p - q_1)^2)^{\alpha/2} (1+(p - q_2)^2)^{\alpha/2} (1+(p')^2)^{\alpha} \\ &\hspace{3cm} \times \psi_j (p - q_1, q_1 + \kappa_1) \overline{\psi}_j  (p', \kappa_1) \; \psi_i (p - q_2, q_2 + \kappa_2) \overline{\psi}_i (p', \kappa_2) \\ \leq \; &  \int \rd p \, \rd p' \, \rd q_1 \rd q_2 \rd \kappa_1 \rd \kappa_2 \, (1+(p')^2)^{\alpha} \;\\  &\hspace{.5cm} \times \left(  \frac{(1+(p - q_1)^2) (1+ (q_1 + \kappa_1)^2) (1+\kappa_2^2)}{(1+(p - q_2)^2)^{1-\alpha} (1+ (q_2 + \kappa_2)^2) (1+\kappa_1^2)} \left|\psi_j (p - q_1, q_1 + \kappa_1) \right|^2  \; \left|\psi_i (p', \kappa_2)\right|^2 \right. \\ & \hspace{1cm} \left. + \frac{(1+(p-q_2)^2) (1+(q_2 + \kappa_2)^2) (1+\kappa_1^2)}{(1+(p-q_1)^2)^{1-\alpha} (1+(q_1 + \kappa_1)^2) (1+\kappa_2^2)} \left|\psi_i (p-q_2, q_2 + \kappa_2)\right|^2  \; \left|\psi_j (p', \kappa_1)\right|^2 \right)
\end{split}
\end{equation}
We consider the first term in the parenthesis; the second one can be handled similarly. Performing the integration over $q_2$ we find
\begin{equation}\label{eq:int2}
\begin{split}
\int \rd p \, \rd &p' \, \rd q_1 \rd q_2 \rd \kappa_1 \rd \kappa_2 \, (1+(p')^2)^{\alpha} \;  \\ &\hspace{1cm} \times \frac{(1+(p-q_1)^2) (1+(q_1 + \kappa_1)^2) (1+\kappa_2^2)}{(1+(p-q_2)^2)^{1-\alpha} (1+(q_2 + \kappa_2)^2) (1+\kappa_1^2)} \; \left|\psi_j (p-q_1, q_1 + \kappa_1)\right|^2  \; \left|\psi_i (p', \kappa_2)\right|^2 \\ \leq \; & \int \rd p \, \rd p' \, \rd q_1 \rd \kappa_1 \rd \kappa_2 \, (1+(p')^2)^{\alpha} \\ &\hspace{1cm} \times  \frac{(1+(p-q_1)^2) (1+(q_1 + \kappa_1)^2) (1+\kappa_2^2)}{(1+(p + \kappa_2)^2)^{(1-\alpha)/2} (1+\kappa_1^2)}\; \left|\psi_j (p- q_1, q_1 + \kappa_1)\right|^2  \; \left|\psi_i (p', \kappa_2)\right|^2
\end{split}
\end{equation}
where we used that
\[ \int \rd q_2 \frac{1}{(1+(p-q_2)^2)^{1-\alpha} (1+(q_2 + \kappa_2)^2)} \leq \frac{C}{(1+(p + \kappa_2)^2)^{(1-\alpha)/2}} \] for all $\alpha <1$. {F}rom (\ref{eq:int2}), we obtain (shifting the integration variables appropriately)
\begin{equation*}
\begin{split}
\int \rd p \, \rd &p' \, \rd q_1 \rd q_2 \rd \kappa_1 \rd \kappa_2 \, (1+(p')^2)^{\alpha} \\ &\hspace{1cm} \times   \frac{(1+(p-q_1)^2) (1+(q_1 + \kappa_1)^2) (1+\kappa_2^2)}{(1+(p-q_2)^2)^{1-\alpha} (1+(q_2 + \kappa_2)^2) (1+\kappa_1^2)} \; \left|\psi_j (p- q_1, q_1 + \kappa_1)\right|^2  \; \left|\psi_i (p', \kappa_2)\right|^2 \\ \leq \; & \int \rd p \, \rd p' \, \rd q_1 \rd \kappa_1 \rd \kappa_2  \\ &\hspace{1cm} \times   \frac{(1+p^2) (1+ q_1^2) (1+\kappa_2^2)(1+(p')^2)^{\alpha}}{(1+(p + q_1 + \kappa_2 -\kappa_1)^2)^{(1-\alpha)/2} (1+\kappa_1^2)} \; \left|\psi_j (p, q_1)\right|^2  \; \left|\psi_i (p', \kappa_2)\right|^2 \\
\leq \; & C_{\alpha} \left( \int \rd p_1 \rd p_2 (1+p_1^2)(1+p_2^2) \, |\psi_j (p_1, p_2)|^2 \right) \left( \int \rd p_1 \rd p_2 (1+p_1^2)(1+p_2^2) \, |\psi_i (p_1, p_2)|^2 \right)
\end{split}
\end{equation*}
where we put
\[ C_{\alpha} = \sup_{P\in \bR^2} \int \rd \kappa_1 \; \frac{1}{(1+\kappa_1^2) (1+ (P-\kappa_1)^2)^{(1-\alpha)/2}} < \infty \] for all $\alpha <1$. The second term in the parenthesis on the r.h.s. of (\ref{eq:int3}) can be bounded similarly. Also the other eight  contributions arising from (\ref{eq:bd22}) can be controlled in a similar way. Therefore, from (\ref{eq:int}), and recalling (\ref{eq:gamma}) we obtain that
\begin{equation*}
\begin{split}
\| S^{(1,\alpha)} B_{1,2} \gamma_{\infty,t}^{(2)} \|_{L^2 (\Lambda \times \Lambda)}^2  \leq \; &C \, \sum_{i,j} \lambda_i \lambda_j \left( \int \rd p_1 \rd p_2 (1+p_1^2)(1+p_2^2) \, |\psi_j (p_1, p_2)|^2 \right) \\ &\hspace{3cm} \times  \left( \int \rd p_1 \rd p_2 (1+p_1^2)(1+p_2^2) \, |\psi_i (p_1, p_2)|^2 \right)
\\ \leq \; & C  \left( \int \rd p_1 \rd p_2 (1+p_1^2)(1+p_2^2) \, \gamma (p_1, p_2 ; p_1,p_2) \right)^2 \\ = \; & C \left( \tr \; (1- \Delta_1) (1-\Delta_2) \gamma^{(2)} \right)^2 \, .
\end{split}
\end{equation*}
\end{proof}

\section{Convergence to the infinite hierarchy}
\setcounter{equation}{0}

\begin{theorem}\label{thm:convergence}
Suppose that $\wt\psi_N$ is defined as in (\ref{eq:wtpsi}), let $\wt\psi_{N,t} = e^{-iH_Nt} \wt\psi_N$ and denote by $\wt \gamma^{(k)}_{N,t}$ the $k$-particle marginal density associated with $\wt \psi_{N,t}$. Suppose that $\Gamma_{\infty,t} = \{ \gamma^{(k)}_{\infty,t} \}_{k \geq 1} \in
\bigoplus_{k \geq 1} C([0,T] , \cL_k^1)$ is a limit point of $\wt \Gamma_{N,t} =
\{ \wt \gamma_{N,t}^{(k)} \}_{k =1}^N$ with respect to the product topology $\tau_{\text{prod}}$ defined in Section \ref{sec:proof}. Then $\Gamma_{\infty,t}$ is a solution to the infinite hierarchy
\begin{equation}\label{eq:conv}
\gamma^{(k)}_{\infty,t} = \cU^{(k)} (t) \gamma^{(k)}_{\infty,0} -
i b_0 \sum_{j=1}^k \int_0^t \rd s \, \cU^{(k)} (t-s) \tr_{k+1}
\left[ \delta (x_j - x_{k+1}), \gamma_{\infty,s}^{(k+1)} \right]
\end{equation}
with initial data $\gamma_{\infty,0}^{(k)} = |\ph \rangle \langle \ph|^{\otimes k}$.
Here $\cU^{(k)} (t)$ denotes the free evolution of $k$ particles defined in~(\ref{eq:free}).
\end{theorem}

\begin{proof}
Fix $k \geq 1$. Passing to an appropriate subsequence, we can assume that,
 for every $J^{(k)} \in \cK_k$,
\begin{equation}\label{eq:conv-0}
\sup_{t \in [0,T]} \, \tr \; J^{(k)} \, \left( \wt\gamma_{N,t}^{(k)} -
\gamma_{\infty,t}^{(k)} \right) \to 0 \qquad \text{as } N \to \infty\,.
\end{equation}
We will prove (\ref{eq:conv}) by testing the limit point against a certain class of observables,  dense in $\cK_k$. More precisely, it is enough to show that, for an arbitrary $J^{(k)} \in \cK_k$ with $\tri J^{(k)} \tri < \infty$,
\begin{equation}\label{eq:conv-1}
\tr \, J^{(k)}  \gamma_{\infty,0}^{(k)} = \tr \, J^{(k)} |\ph
\rangle \langle \ph|^{\otimes k}
\end{equation}
and
\begin{equation}\label{eq:conv-2}
\begin{split}
\tr \; J^{(k)}  \gamma_{\infty,t}^{(k)} = \tr \; J^{(k)} \cU^{(k)}
(t)  \gamma_{\infty,0}^{(k)} - i b_0 \sum_{j=1}^k \int_0^t \rd s
\tr \, J^{(k)} \cU^{(k)} (t-s) \left[ \delta (x_j -x_{k+1}),
\gamma^{(k+1)}_{\infty,s} \right]\,.
\end{split}
\end{equation}

Fix now $J^{(k)} \in \cK_k$ such that $\tri J^{(k)} \tri < \infty$ (recall the definition of the norm $\tri . \tri$ from (\ref{eq:tri})). Eq. (\ref{eq:conv-1}) follows immediately from (\ref{eq:conv-0}). To prove (\ref{eq:conv-2}), we use the BBGKY hierarchy (\ref{eq:BBGKY1}), rewritten in integral form as
\begin{equation}\label{eq:conv-4}
\begin{split}
\tr \; J^{(k)} \, \wt\gamma_{N,t}^{(k)} = \; & \tr \; J^{(k)} \,
\cU^{(k)} (t) \wt\gamma_{N,0}^{(k)} - \frac{i}{N} \sum_{i<j}^k \int_0^t  \rd s \,
\tr \; J^{(k)} \, \cU^{(k)}
(t-s) [ N^{2\beta} V (N^{\beta}(x_i -x_j)), \wt \gamma_{N,s}^{(k)} ] \\
& - i \left(1-\frac{k}{N} \right) \sum_{j=1}^k \int_0^t \rd s \, \tr J^{(k)} \cU^{(k)}
(t-s) [ N^{2\beta} V (N^{\beta}(x_j -x_{k+1})) , \wt \gamma_{N,s}^{(k+1)} ] \,.
\end{split}
\end{equation}
Since, by (\ref{eq:conv-0}), the term on the l.h.s. of (\ref{eq:conv-4}) and the first term on the r.h.s. of (\ref{eq:conv-4}) converge to the term on the l.h.s. of (\ref{eq:conv-2}) and, respectively, to the first term on the r.h.s. of (\ref{eq:conv-2}) (using the assumption that $\tri J^{(k)} \tri < \infty$), and since the second term on the r.h.s. of (\ref{eq:conv-4}) vanishes as $N \to\infty$ (by a simple computation similar to (\ref{eq:model})), it is enough to prove that the last term on the r.h.s. of (\ref{eq:conv-4}) converges, as $N \to \infty$, to the last term on the r.h.s. of (\ref{eq:conv-2}). The contribution proportional to $k/N$ in the last term on the r.h.s. of (\ref{eq:conv-4}) can be shown to vanish as $N \to \infty$ (again, with an argument similar to (\ref{eq:model})). Moreover, the two terms arising from the commutator can be handled similarly. Therefore, we have to prove that, for fixed $T$, $k$ and $J^{(k)}$,
\begin{equation}\label{eq:conv-claim}
\sup_{s \leq t \leq T} \left| \tr J^{(k)} \cU^{(k)}(t-s) \left( N^{2\beta} V (N^{\beta}(x_j -x_{k+1})) \wt \gamma_{N,s}^{(k+1)} - b_0 \delta (x_j -x_{k+1}) \gamma_{\infty,s}^{(k+1)} \right) \right| \to 0
\end{equation}
as $N \to \infty$. To this end, we choose a probability measure $h \in L^1 (\Lambda)$ with $h \geq 0$ and $\int h = 1$, and for arbitrary $\alpha >0$ we define $h_{\alpha} (x) = \alpha^{-2} h(x/\alpha)$. Then we observe that
\begin{equation}\label{eq:4terms}
\begin{split}
\Big| \tr J^{(k)} \cU^{(k)}(t-s) &\left( N^{2\beta} V (N^{\beta}(x_j -x_{k+1})) \wt \gamma_{N,s}^{(k+1)} - b_0 \delta (x_j -x_{k+1}) \gamma_{\infty,s}^{(k+1)} \right) \Big| \\ \leq \; & \Big|
\tr \,  J_{s-t}^{(k)} \, \left( N^{2\beta} V (N^{\beta}(x_j -x_{k+1})) - b_0 \delta (x_j - x_{k+1}) \right) \wt \gamma_{N,s}^{(k+1)} \Big| \\ &+ b_0 \Big|
\tr \, J_{s-t}^{(k)} \, \left(\delta (x_j - x_{k+1}) - h_{\alpha} (x_j -x_{k+1}) \right) \wt\gamma_{N,s}^{(k+1)} \Big| \\ &+ b_0 \Big|
\tr \, J_{s-t}^{(k)} \,  h_{\alpha} (x_j -x_{k+1}) \left( \wt\gamma_{N,s}^{(k+1)} -\gamma_{\infty,s}^{(k+1)}\right) \Big| \\
&+ b_0 \Big|
\tr \, J_{s-t}^{(k)} \,  \left(h_{\alpha} (x_j - x_{k+1}) - \delta (x_j -x_{k+1}) \right) \gamma_{\infty,s}^{(k+1)} \Big|
\end{split}
\end{equation}
where we introduced the notation $J^{(k)}_{t} = \cU^{(k)} (t) J^{(k)}$. The first term on the r.h.s. of the last equation converges to zero as $N \to\infty$, by Lemma \ref{lm:poincare} and by the a-priori bounds (\ref{eq:apriNk}). The second and fourth term on the r.h.s. of the last equation converge to zero, as $\alpha \to 0$, uniformly in $N$ (again by Lemma \ref{lm:poincare}, once combined with (\ref{eq:apriNk}) and once with (\ref{eq:aprik-inf})). The third term on the r.h.s. of the last equation converges to zero as $N \to \infty$, for every fixed $\alpha$. Thus, taking first the limit $N \to \infty$, and then letting $\alpha \to 0$, we obtain (\ref{eq:conv-claim}). To prove that the third term on the r.h.s. of (\ref{eq:4terms}) converges to zero as $N \to \infty$, for every fixed $\alpha >0$, note that, for arbitrary $\eps >0$,
\begin{equation}\label{eq:cmp}
\begin{split}
\Big| \tr \, J_{s-t}^{(k)} \,  h_{\alpha} &(x_j -x_{k+1}) \left( \wt\gamma_{N,s}^{(k+1)} -\gamma_{\infty,s}^{(k+1)}\right)\Big| \\ \leq \; & \Big| \tr \, J_{s-t}^{(k)} \,  h_{\alpha} (x_j -x_{k+1}) \frac{1}{1+\eps S_{k+1}} \left( \wt\gamma_{N,s}^{(k+1)} -\gamma_{\infty,s}^{(k+1)}\right)\Big| \\ &+\Big| \tr \, J_{s-t}^{(k)} \,  h_{\alpha} (x_j -x_{k+1}) \frac{\eps S_{k+1}}{1+\eps S_{k+1}} \left( \wt\gamma_{N,s}^{(k+1)} -\gamma_{\infty,s}^{(k+1)}\right)\Big|
\end{split}
\end{equation}
The first term converges to zero, as $N \to \infty$ by (\ref{eq:conv-0}), for every fixed $\eps >0$ (because the operator $J_{s-t}^{(k)} h_{\alpha} (x_j -x_{k+1}) (1+\eps S_{k+1})^{-1}$ is compact for every $\eps >0$). The second term on the r.h.s. of (\ref{eq:cmp}) converges to zero as $\eps \to 0$, uniformly in $N$ (making use of (\ref{eq:apriNk}) and (\ref{eq:aprik-inf})).
\end{proof}

\section{Uniqueness of the solution to the infinite hierarchy}\label{sec:unique}
\setcounter{equation}{0}

\subsection{The case $\Lambda = \bR^2$}
\label{sec:conti}

In this subsection we assume that $\Lambda = \bR^2$.
\begin{theorem}\label{thm:unique-R2} Fix $\alpha > 1/2$, $T >0$ and $\Gamma = \{ \gamma^{(k)} \}_{k \geq 1} \in \bigoplus \cL^1_k$. There exists at most one solution $\Gamma_t = \{ \gamma^{(k)}_t \}_{k\geq 1} \in \bigoplus_{k\geq 1} C([0,T], \cL_k^1 )$ to the infinite hierarchy
\begin{equation}\label{eq:infhie2} \gamma^{(k)}_t = \cU^{(k)} (t) \gamma^{(k)} -i b_0 \sum_{j=1}^k \int_0^t \rd s\, \cU^{(k)} (t-s) B_{j,k+1} \gamma_s^{(k+1)} \end{equation} with $\Gamma_{t=0} = \Gamma$ and such that
\[ \Big\| S^{(k,\alpha)} B_{j,k+1} \gamma_s^{(k+1)} \Big\|_{L^2 (\bR^{2k} \times \bR^{2k})} \leq C^k \] for all $k \geq 1$. Here the map $B_{j,k+1}$, for $j=1,\dots,k$ has been defined in (\ref{eq:Bk}), the free evolution $\cU^{(k)} (t)$ in (\ref{eq:free}), and $S^{(k,\alpha)} = \prod_{j=1}^k (1-\Delta_{x_j})^{\alpha/2} (1-\Delta_{x'_j})^{\alpha/2}$.
\end{theorem}

A solution to (\ref{eq:infhie2}) can be expanded in a Duhamel type series
\begin{equation}\label{eq:duh}
\begin{split}
\gamma^{(k)}_t = \; &\cU^{(k)} (t) \gamma^{(k)}_0 + \sum_{m=1}^{n-1} \eta^{(k)}_{m,t} + \xi_{n,t}^{(k)}
\end{split}
\end{equation}
where
\begin{equation}
\begin{split}
\eta^{(k)}_{m,t} = \; &\sum_{j_1=1}^k \dots \sum_{j_m=1}^{k+m-1} \int_0^t \rd s_1 \dots \int_0^{s_{m-1}} \rd s_m \,
\cU^{(k)} (t-s_1) B_{j_1,k+1} \cU^{(k+1)} (s_1 -s_2) \dots \\ &\hspace{5cm} \dots \times B_{j_m,k+m} \cU^{(k+m)} (s_m) \gamma^{(k+m)}
\end{split}
\end{equation}
and
\begin{equation}\label{eq:xik}
\begin{split}
\xi^{(k)}_{n,t} = \; &\sum_{j_1=1}^k \dots \sum_{j_n=1}^{k+n-1} \int_0^t \rd s_1 \dots \int_0^{s_{n-1}} \rd s_n \,
\cU^{(k)} (t-s_1) B_{j_1,k+1} \cU^{(k+1)} (s_1 -s_2) \dots B_{j_m,k+m}  \gamma^{(k+n)}_{s_n}\,.
\end{split}
\end{equation}
To prove the uniqueness of the solution it is enough to show that the error term $\xi_{n,t}^{(k)}$ (all other terms $\eta^{(k)}_{m,t}$ only depend on the initial data) converges to zero, as the order $n$ of the expansion tends to infinity, for small but fixed time $t >0$ (uniqueness for all times follows then by repeating the argument). To this end, we follow the technique developed, in the three dimensional setting, by Klainerman and Machedon; see \cite{KM}. This method relies on two ingredients; the first ingredient is an expansion of the error term (\ref{eq:xik}) in a sum of less than $C^n$ contributions, obtained by an appropriate recombination of the terms associated with different indices $j_1, \dots ,j_n$ in (\ref{eq:xik}) (originally, the sums over $j_1, \dots, j_n$ involve factorially many summands). This reorganization of the terms in (\ref{eq:xik}) was first introduced in \cite{ESY2}, and can be interpreted as an expansion in Feynman diagrams. The second ingredient, which was the main novelty of \cite{KM}, is a space-time estimate, which is then applied recursively to show that all the terms in the expansion are bounded. The expansion used by Klainerman and Machedon (see Section 3 of \cite{KM}, or the diagrammatic expansion developed in \cite{ESY2}, Section 9) can be used with no changes also in the two-dimensional setting we are considering here. Therefore, to complete the proof of Theorem \ref{thm:unique-R2}, we only have to show the following proposition, which replaces, in the argument of Klainerman and Machedon, the space-time estimate given by \cite[Theorem 1.3]{KM}.
\begin{proposition}
Let $S^{(k,\alpha)} = \prod_{j=1}^k (1-\Delta_{x_j})^{\alpha/2} (1-\Delta_{x'_j})^{\alpha/2}$ and let $\cU^{(k)} (t)$ be the free evolution of $k$-particle defined in (\ref{eq:free}). Then, for every $\alpha >1/2$ and for every $j=1, \dots, k$, we have
\begin{equation}
\left\| S^{(k,\alpha)} B_{j,k+1} \cU^{(k+1)} (t) \gamma^{(k+1)} \right\|_{L^2 (\bR \times \bR^{2k} \times \bR^{2k})} \leq C \, \left\| S^{(k+1,\alpha)} \gamma^{(k+1)} \right\|_{L^2 (\bR^{2k} \times \bR^{2k})} \,.
\end{equation}
\end{proposition}
{\it Remark.} Note that in \cite{KM}, the operator $S^{(k,\alpha)}$ is replaced by $R^{(k)} = \prod_{j=1}^k |\nabla_{x_j}|\, |\nabla_{x'_j}|$, but this difference does not affect the rest of the argument. In the two-dimensional setting, the proof of the uniqueness is clearly simpler than in the three dimensional setting because the singularity of the delta-function is less critical.

\begin{proof}
We can apply the same ideas used by Klainerman and Machedon. We set $j=1$, and note that, in Fourier space,
\begin{equation}\begin{split} \Big( &\widehat{S^{(k,\alpha)} B_{1,k+1} \cU^{(k+1)} (.) \gamma^{(k+1)}} \Big) (\tau; \bp_k ; \bp'_k ) \\ &\hspace{1cm} = \prod_{j=1}^k (1 + p_j^2)^{\alpha/2} (1+ (p'_j)^2)^{\alpha/2} \int \rd q \rd q' \; \delta (\tau + |\bp_k|^2 + |p_1 - q + q'|^2 - |\bp'_k|^2 - |q'|^2) \\ &\hspace{7cm} \times \gamma^{(k+1)} (p_1 - q + q' , p_2 ,\dots , p_k, q ; p'_1, \dots , p'_k , q') \end{split}
\end{equation}
Therefore
\begin{equation}
\begin{split}
 \Big\| S^{(k,\alpha)}& B_{j,k+1} \cU^{(k+1)} (t) \gamma^{(k+1)} \Big\|_{L^2 (\bR \times \bR^{2k} \times \bR^{2k})} \\  = \; & \int \rd \tau \rd \bp_k \rd \bp'_k \; \Big| \int \rd q \rd q'  \; \delta (\tau + |\bp_k|^2 + |p_1 - q + q'|^2 - |\bp'_k|^2 - |q'|^2)
 \\ &\hspace{1.5cm} \times \prod_{j=1}^k (1 + p_j^2)^{\alpha/2} (1+ (p'_j)^2)^{\alpha/2}
 \gamma^{(k+1)} (p_1 - q + q' , p_2 ,\dots , p_k, q ; p'_1, \dots , p'_k , q') \Big|^2
\end{split}
\end{equation}
With a weighted Schwarz inequality, we obtain
\begin{equation}
\begin{split}
\Big\| S^{(k,\alpha)}& B_{j,k+1} \cU^{(k+1)} (t) \gamma^{(k+1)} \Big\|_{L^2 (\bR \times \bR^{2k} \times \bR^{2k})} \\
\leq \; & \int \rd \tau \rd \bp_k \rd \bp'_k \; \left( \int \rd q \rd q' \,  \frac{(1+p_1^2)^{\alpha} \; \delta (\tau + |\bp_k|^2 + |p_1 - q + q'|^2 - |\bp'_k|^2 - |q'|^2)}{(1+(p_1 - q + q')^2)^{\alpha} (1+ q^2)^{\alpha} (1+(q')^2)^{\alpha}} \right) \\ & \times  \Big( \int \rd q \rd q'  \, \delta (\tau + |\bp_k|^2 + |p_1 - q + q'|^2 - |\bp'_k|^2 - |q'|^2) \\ &\hspace{2cm} \times  (1+ (p_1 - q +q')^2)^{\alpha} (1+ q^2)^{\alpha} (1+ (q')^2)^{\alpha} \prod_{j=2}^k (1 + p_j^2)^{\alpha}   \prod_{j=1}^k (1+ (p'_j)^2)^{\alpha}  \\ &\hspace{2cm} \times  |\gamma^{(k+1)} (p_1 - q + q' , p_2 ,\dots , p_k, q ; p'_1, \dots , p'_k , q')|^2  \Big)\\
 \leq \; & \Big\| S^{(k+1,\alpha)} \gamma^{(k+1)} \Big\|_{L^2 (\bR^{2k} \times \bR^{2k})} \\ &\hspace{2cm} \times \sup_{\tau, \bp_k, \bp'_k} \int \rd q \rd q' \, \frac{(1+p_1^2)^{\alpha}\, \delta (\tau + |\bp_k|^2 + |p_1 - q + q'|^2 - |\bp'_k|^2 - |q'|^2) }{(1+(p_1 - q + q')^2)^{\alpha} (1+ q^2)^{\alpha} (1+(q')^2)^{\alpha}}
\end{split}
\end{equation}
and the proposition follows from Lemma \ref{lm:delta}.
\end{proof}

\begin{lemma}\label{lm:delta}
For all $\alpha > 1/2$, we have
\begin{equation}
I =\sup_{\tau,p}  \int_{\bR^2 \times \bR^2}  \rd q \rd q' \, \frac{(1+p^2)^{\alpha} \, \delta \left( \tau + |p + q - q'|^2 + |q|^2 - |q'|^2 \right)}{(1+ (p + q - q')^2)^{\alpha} (1+ q^2)^{\alpha} (1+ (q')^2 )^{\alpha}} < \infty \, .
\end{equation}
\end{lemma}
\begin{proof}
{F}rom
\[ (1+p^2)^{\alpha} \leq (1+ (p+q-q')^2)^{\alpha} + (1+q^2)^{\alpha} + (1+(q')^2)^{\alpha} \]
and using the invariance of the integral with respect to the shift $q \to p+q-q'$,
we obtain
\begin{equation}\label{eq:I1I2}
\begin{split}
I \leq \; & 2 \sup_{\tau,p} \int_{\bR^2 \times \bR^2}  \rd q \rd q' \, \delta \left( \tau + |p + q - q'|^2 + |q|^2 - |q'|^2 \right) \, \frac{1}{(1+ q^2)^{\alpha} (1+ (q')^2 )^{\alpha}} \\ &+\sup_{\tau,p} \int_{\bR^2 \times \bR^2}  \rd q \rd q' \, \delta \left( \tau - |p + q - q'|^2 + |q|^2 + |q'|^2 \right) \, \frac{1}{(1+ q^2)^{\alpha} (1+ (q')^2 )^{\alpha}} \\ =\; & I_1 + I_2
\end{split}
\end{equation}
To bound the first term on the r.h.s., we fix $q$ and rotate $q'=(q'_1, q'_2)$ so that $q' \cdot (p+q) = q'_1 |p+q|$. Then
\[ \tau + (p+q-q')^2 + q^2 - (q')^2 = \tau +(p+q)^2 +q^2 -2|p+q| q'_1 \]
and
\begin{equation}
\begin{split}
I_1 = \; & \sup_{p, \tau} \int \frac{\rd q}{|p+q| (1+q^2)^{\alpha}} \int_{-\infty}^{\infty}
\rd q'_2 \frac{1}{\left(1+(q'_2)^2 + \left( \frac{\tau + (p+q)^2 + q^2}{2|p+q|} \right)^2 \right)^{\alpha}} \\ \leq \; & \sup_p \int \frac{\rd q}{|p+q| (1+q^2)^{\alpha}} \int_{-\infty}^{\infty}
\frac{\rd q'_2}{(1+(q'_2)^2)^{\alpha}} \\ \leq \; & C
\end{split}
\end{equation}
because
\[ \sup_{p \in \bR^2} \int  \frac{\rd q}{|p-q| (1 + q^2)^{\alpha}} < \infty \qquad \text{and } \int_{-\infty}^{\infty} \frac{\rd x}{(1+x^2)^{\alpha}} < \infty \]
for all $\alpha >1/2$. The second term on the r.h.s. of (\ref{eq:I1I2}) can be bounded similarly.
\end{proof}

\subsection{The case $\Lambda = [-L,L]^{\times2}$}
\label{sec:per}

\begin{theorem}\label{thm:unique-box}
Fix $\alpha > 1/2$, $T > 0$, and consider initial data $\Gamma = \{\gamma^{(k)}\}_{k \geq 1} \in \bigoplus \cL^1_k$. Then there exists at most one solution $\Gamma_t = \{ \gamma^{(k)}_t \}_{k\geq 1} \in \bigoplus_{k\geq 1} C([0,T], \cL_k^1 )$ of the infinite (Gross-Pitaevskii) hierarchy on the domain $\Lambda = [-L,L]^{\times2}$:
\begin{equation}\label{GPH}
 \gamma^{(k)}_t = \cU^{(k)} (t) \gamma^{(k)} -i b_0 \sum_{j=1}^k \int_0^t \rd s\, \cU^{(k)} (t-s) B_{j,k+1} \gamma_s^{(k+1)}
\end{equation}
with $\Gamma_{t=0} = \Gamma$ and such that
\[ \Big\| S^{(k,\alpha)} B_{j,k+1} \gamma_s^{(k+1)} \Big\|_{L^2 (\Lambda^{k} \times \Lambda^{k})} \leq C^k \] for all $k \geq 1$. Here the map $B_{j,k+1}$, for $j=1,\dots,k$, is the collision operator defined in (\ref{eq:Bk}), the free evolution $\cU^{(k)} (t)$ in (\ref{eq:free}), and $S^{(k,\alpha)} = \prod_{j=1}^k (1-\Delta_{x_j})^{\alpha/2} (1-\Delta_{x'_j})^{\alpha/2}$.
\end{theorem}

The proof follows the general outline of Klainerman and Machedon's in \cite{KM}. The main novelty is the use of some number-theoretic estimates, like the Gauss lemma \cite{BP, H, IR}, that were  first used in a PDE context in \cite{B1} (see also \cite{DPST}). After a Duhamel expansion argument as in the previous subsection, we need to show the following:

\begin{proposition}
Let $\Gamma_t = \{ \gamma^{(k)}_t \} = \{ \cU^{(k)} (t) \gamma^{(k)} \} $ satisfy the homogeneous infinite hierarchy. Then for every  $\alpha >1/2$ and for every $j=1, \dots, k$, we have
\begin{equation}\label{mainestimate}
 \Vert S^{(k, \alpha)} B_{j, k+1} \cU^{(k+1)} (t) \gamma^{(k+1)} \Vert_{L^2(\mathbb{R} \times \Lambda^{k} \times \Lambda^{k})} \leq C \Vert S^{(k+1, \alpha)} \gamma^{(k+1)} \Vert_{L^2(\Lambda^{(k+1)} \times \Lambda^{(k+1)})}.
\end{equation}
\end{proposition}

\begin{proof}
Without loss of generality we set $j=1$ (which will be suppressed in the sequel), and $\Lambda = [-1,1]^{{\times}2}$. We also assume that $\alpha = 1$, but one can easily see that the argument also works as long as $\alpha>1/2$. We use the same initial approach as Klainerman and Machedon and observe that, by Plancherel's theorem, equation \eqref{mainestimate} is equivalent to:
\begin{equation}\label{Plancherel}
 \Vert I_k [S^{(k+1)}\gamma^{(k+1)}_t] \Vert_{L^2_\tau(\mathbb{R}) \ell^2(\bZ^{2k} \times \bZ^{2k})} \leq C \Vert \widehat{S^{(k+1)}\gamma^{(k+1)}_0} \Vert_{\ell^2(\bZ^{2(k+1)} \times \bZ^{2(k+1)})},
\end{equation}
where
\begin{equation}
 I_k [f](\tau, \bn_k, \bn'_k) = \sum_{n'_{k+1}} \sum_{n_{k+1}} \delta(\dots) \frac{ \langle n_1 \rangle \hat{f} (n_1 - n_{k+1} - n'_{k+1}, \dots, \bn'_{k+1}) }{ \langle n_1 - n_{k+1} - n'_{k+1} \rangle \langle n_{k+1} \rangle \langle n'_{k+1} \rangle },
\end{equation}
and
\begin{equation}
\delta( \dots ) := \delta ( \tau + \abs{ n_1 - n_{k+1} - n'_{k+1} }^2 + \abs{ \bn_{k+1} }^2 - \abs{ n_1 }^2 - \abs{ \bn'_{k+1} }^2),
\end{equation}
and $\langle n \rangle := (1 + \abs{n}^2)^{1/2}.$

Then by Cauchy-Schwarz,
\begin{equation}\label{secondterm}
\begin{split}
\abs{ I_k [f] }^2 \leq & \sum_{n'_{k+1}} \sum_{n_{k+1}} \delta(\dots) \abs{ \hat{f} (n_1 - n_{k+1} - n'_{k+1}, \dots, \bn'_{k+1}) }^2 \\
& \times \sum_{n'_{k+1}} \sum_{n_{k+1}} \delta(\dots) \frac{ \langle n_1 \rangle^2 }{ \langle n_1 - n_{k+1} - n'_{k+1} \rangle^2 \langle n_{k+1} \rangle^2 \langle n'_{k+1} \rangle^2 }
\end{split}
\end{equation}

If the second factor on the right-hand side of \eqref{secondterm} is bounded, then:
\begin{equation}
 \begin{split}
\Vert I_k [f] \Vert^2_{L^2_\tau(\mathbb{R}) \ell^2(\bZ^{2k} \times \bZ^{2k})} & \leq  C^2 \int \sum_{n'_{k+1}} \sum_{n_{k+1}} \delta(\dots) \abs{ \hat{f} (n_1 - n_{k+1} - n'_{k+1}, \dots, \bn'_{k+1}) }^2 d\tau \\ &
\leq C^2 \Vert \hat{f} \Vert^2_{\ell^2}.
 \end{split}
\end{equation}

Hence \eqref{Plancherel} and the proposition follow from Lemma \ref{lm:delta2}.

\end{proof}

\begin{lemma}\label{lm:delta2} For all $\tau \in \bR$ and $p \in \bZ^2$,
\begin{equation}\label{eq:delta2}
 \sum_{ n, m \in \bZ^2 } \frac{ \delta ( \tau + \abs{ p-n-m }^2 + \abs{ n }^2 - \abs{ m }^2 ) \langle p \rangle ^2 }{ \langle{ p-n-m }\rangle^2 \langle n \rangle ^2 \langle m \rangle ^2 } \leq C < \infty \, .
\end{equation}
\end{lemma}

\begin{proof}
Since (\ref{eq:delta2}) is not symmetric with respect to $n$ and $m$, we consider the two cases, $ \abs{n} \ll \abs{m}$ and $\abs{n} \gtrsim \abs{m}$.

\subsubsection{Case I: $ \abs{n} \ll \abs{m} $}

We decompose the sum in (\ref{eq:delta2}) as follows:
\begin{equation}
 \sum_{\substack{n,m \in \bZ^2 \\ \abs{n} \ll \abs{m} }} \delta(***) = \sum_{i > j \geq 0}  \sum_{\abs{n} \sim 2^j} \sum_{\abs{m} \sim 2^i} \delta(***) = \sum_{i>j \geq 0} \sum_{\abs{n} \sim 2^j} \# S_l,
\end{equation}
where $S_l$ is defined to be, for fixed $n$:
\begin{equation}
 S_l = S_{n,i,\tau,p} := \{ \abs{m} \sim 2^i : \abs{p-n-m}^2 + \abs{n}^2 - \abs{m}^2 = -\tau \}.
\end{equation}

In order to compute $ \# S_l $, we fix $m_0 \in S_l$ and count the number of $\ell \in \bZ^2$ such that $m_0+\ell \in S_l$. By definition, two equations must be satisfied for $m_0$ and $m_0+\ell$ to be in $ S_l$:
\begin{equation}
 -\tau  = \abs{p-n-m_0}^2 + \abs{n}^2 - \abs{m_0}^2,
\end{equation}
and
\begin{equation} \begin{split}
 -\tau  &= \abs{p-n-(m_0 + \ell)}^2 + \abs{n}^2 - \abs{m_0 + \ell}^2 \\
& = \abs{p-n-m_0}^2 + \abs{\ell}^2 - 2\ell \cdot (p-n-m_0) + \abs{n}^2  \\
& \quad \quad \quad - \abs{m_0}^2 - \abs{\ell}^2 - 2\ell \cdot m_0 \\
& = \abs{p-n-m_0}^2 - 2\ell \cdot (p-n) + \abs{n}^2 - \abs{m_0}^2.
\end{split} \end{equation}
Subtracting the first equation from the second gives the following linear equation for $\ell$:
\begin{equation}\label{elllineareqn}
 \ell \cdot (p-n) = 0.
\end{equation}
The continuous counterpart of this equation is, for $\bx = (x,y)$:
\begin{equation}\label{contlineareqn}
 \bx \cdot (p-n) = 0.
\end{equation}
And $\abs{m_0}, \abs{m_0+\ell} \sim 2^i$ imply that $\abs{\ell} \lesssim 2^i$.
Thus $ \# S_l $ is the number of lattice points on $T$ (the line defined by \eqref{contlineareqn}) inside $D$ (the disc of radius $2^i$ centered at the origin). (See Figure \ref{CzeroTline}.)
\begin{equation}\label{latticepoints}
\# S_l = \# \{\ell\} + 1 \lesssim 2^i.
\end{equation}

\begin{figure}[htbp]
\centering
\epsfig{file=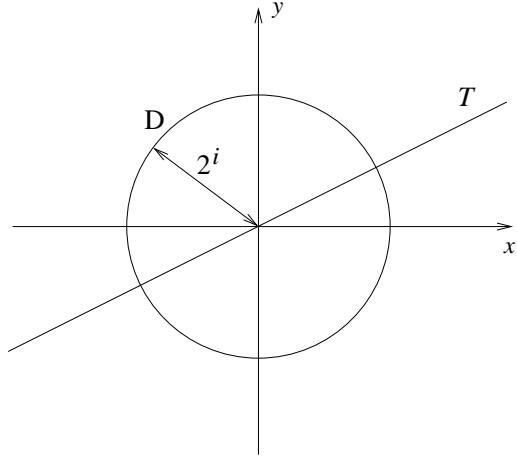, height=2.4in}
\caption{ The graph of $T$, a level set of the linear equation \eqref{contlineareqn}. }
\label{CzeroTline}
\end{figure}

\textbf{Case IA}: If $\abs{n} \ll \abs{m} \ll \abs{p}$, then $\abs{p-n-m} \sim \abs{p}$, and we can cancel $\langle p \rangle^2$ with $\langle p-n-m \rangle^2$ and control the left-hand side of (\ref{eq:delta2}) by a convergent series as follows:
\begin{equation} \begin{split}
  \sum_{\substack{ n, m \in \bZ^2 \\ \abs{n} \ll \abs{m} \ll \abs{p} } } \frac{ \delta ( *** ) \langle p \rangle ^2 }{ \langle{ p-n-m }\rangle^2 \langle n \rangle ^2 \langle m \rangle ^2 }
& \lesssim \sum_{i > j} \sum_{\abs{n} \sim 2^j} \sum_{\abs{m} \sim 2^i} \frac{ \delta(***) }{ \langle n \rangle^2 \langle m \rangle ^2 } \\
 & \lesssim \sum_{j=0}^{\infty} \sum_{\abs{n} \sim 2^j} \frac{1}{\langle n \rangle^2} \sum_{i=j+1}^{\infty} \frac{ \# S_l }{ 2^{2i} } \\
 & \lesssim \sum_{j=0}^{\infty} \sum_{\abs{n} \sim 2^j} \frac{1}{\langle n \rangle^2} \sum_{i=j+1}^{\infty} \frac{ 2^i }{ 2^{2i} } \\
 & \lesssim \sum_{j=0}^{\infty} \sum_{\abs{n} \sim 2^j} \frac{1}{\langle n \rangle^2} \frac{1}{2^{j}} \\
 & \lesssim \sum_{j=0}^{\infty} 2^{2j} \frac{1}{2^{2j}} \frac{1}{2^j} \\
 & = 2.
 \end{split} \end{equation}

\textbf{Case IB}: If $\abs{p} \lesssim \abs{n} \ll \abs{m}$, then we cancel $\langle p \rangle^2$ with $\langle n \rangle^2$ and estimate, for some small, positive $\eps$:
\begin{equation} \begin{split}
  \sum_{\substack{ \abs{p} \lesssim \abs{n} \ll \abs{m}  } } \frac{ \delta ( *** ) \langle p \rangle ^2 }{ \langle{ p-n-m }\rangle^2 \langle n \rangle ^2 \langle m \rangle ^2 } & \lesssim \sum_{n} \sum_m \frac{ \delta(***)  }{ \langle{ p-n-m }\rangle^2 \langle m \rangle ^2 } \\
 & \lesssim \sum_{n} \sum_m \frac{ \delta(***) }{ \langle m \rangle^2 \langle m \rangle^2} \\
 & \lesssim \sum_{n} \sum_m \frac{ \delta(***) }{ \langle m \rangle^{2+\eps} \langle m \rangle^{2-\eps}} \\
& \lesssim \sum_{n}  \frac{1 }{ \langle n \rangle^{2+\eps}} \sum_m \frac{ \delta(***) }{ \langle m \rangle^{2-\eps}} < \infty.
\end{split} \end{equation}
We borrowed an $\eps$ power of $m$ in order to make the first (two-dimensional) sum converge, since $2 + \eps > 2$. The second sum is really one-dimensional due to the restriction, hence converges because $2 - \eps > 1$.

\textbf{Subcase IC1}: If $ \abs{n} \ll \abs{p} \lesssim \abs{m}$ and $\abs{p-m} \gg \abs{n}$, then for some $\eps > 0$:
\begin{equation} \begin{split}
  \sum_{\substack{ \abs{n} \ll \abs{p} \lesssim \abs{m} \\ \abs{p-m} \gg \abs{n} } } \frac{ \delta ( *** ) \langle p \rangle ^2 }{ \langle{ p-n-m }\rangle^2 \langle n \rangle ^2 \langle m \rangle ^2 } & \lesssim \sum_n \sum_m \frac{ \delta(***) }{ \langle{ p-m }\rangle^2 \langle n \rangle ^2 } \\
 & \lesssim \sum_n \sum_m \frac{ \delta(***) }{ \langle{ p-m }\rangle^{2-\eps} \langle n \rangle ^{2+\eps} } \\
 & \lesssim \sum_{n} \frac{1}{\langle n \rangle^{2+\eps}} \sum_{m} \frac{\delta(***)}{\langle p-m \rangle^{2-\eps}} < \infty.
\end{split} \end{equation}
Again the first sum in the last step is two-dimensional; and the second is one-dimensional due to the restriction. (To see this, one just needs to follow the argument presented above, replacing $m$ by $p-m$.)

\textbf{Subcase IC2}: If $ \abs{n} \ll \abs{p} \lesssim \abs{m}$ and $\abs{p-m} \ll \abs{n}$, then we use the fact that $\langle n \rangle^{-2+\eps} \ll \langle p-m \rangle^{-2+\eps}$ and proceed like the previous case:
\begin{equation} \begin{split}
  \sum_{\substack{ \abs{n} \ll \abs{p} \lesssim \abs{m} \\ \abs{p-m} \ll \abs{n} } } \frac{ \delta ( *** ) \langle p \rangle ^2 }{ \langle{ p-n-m }\rangle^2 \langle n \rangle ^2 \langle m \rangle ^2 } & \lesssim \sum_n \sum_m \frac{ \delta(***) }{ \langle{ n }\rangle^2 \langle n \rangle ^2 } \\
 & \lesssim \sum_n \sum_m \frac{ \delta(***) }{ \langle{ n }\rangle^{2-\eps} \langle n \rangle ^{2+\eps} }  \\
 & \lesssim \sum_n \sum_m \frac{ \delta(***) }{ \langle{ p-m }\rangle^{2-\eps} \langle n \rangle ^{2+\eps} }  \\
 & \lesssim \sum_{n} \frac{1}{\langle n \rangle^{2+\eps}} \sum_{m} \frac{\delta(***)}{\langle p-m \rangle^{2-\eps}} < \infty.
\end{split} \end{equation}

\textbf{Subcase IC3}: If $\abs{n} \ll \abs{p} \lesssim \abs{m}$ and $\abs{p-m} \sim \abs{n}$, then we use the change of variables $z := m+n$ in the third step below:
\begin{equation} \begin{split}
  \sum_{\substack{ \abs{n} \ll \abs{p} \lesssim \abs{m} \\ \abs{p-m} \sim \abs{n} } } & \frac{ \delta ( *** ) \langle p \rangle ^2 }{ \langle{ p-n-m }\rangle^2 \langle n \rangle ^2 \langle m \rangle ^2 }  \lesssim \sum_n \sum_m  \frac{ \delta(***) }{ \langle{ p-m-n }\rangle^2 \langle n \rangle ^2 } \\
 & \lesssim \sum_m \sum_n \frac{ \delta(\abs{p-n-m}^2 + \abs{n}^2 - \abs{m}^2 + \tau) }{ \langle p-m-n \rangle^{2} \langle n \rangle^{2} }\\
 & \lesssim \sum_z \sum_{j=0}^{\infty} \sum_{\abs{n} \sim 2^j} \frac{ \delta( \abs{p-z}^2 + \abs{n}^2 - \abs{z-n}^2 + \tau) }{ \langle p-z  \rangle^{2} \langle n \rangle ^{2} } \\
 & \lesssim \sum_z \sum_{j=0}^{\infty} \sum_{\abs{n} \sim 2^j} \frac{ \delta( \abs{p-z}^2 + \abs{n}^2 - \abs{z-n}^2 + \tau) }{ \langle p-z \rangle^{2+\eps} \langle n \rangle ^{2-\eps} } \\
 & \lesssim \sum_z \frac{1}{\langle p-z \rangle^{2+\eps}} \sum_{j=0}^{\infty} \frac{2^j}{2^{(2-\eps)j}} \\
& \lesssim \sum_z \frac{1}{\langle p-z \rangle^{2+\eps}} < \infty.
\end{split} \end{equation}
For the fourth step, we used $\abs{p-z} = \abs{p-n-m} \lesssim \abs{p-m} + \abs{n} \lesssim 2 \abs{n}$ in order to take an $\langle n \rangle^\eps$ and combine it with $\langle p-z \rangle^2$. And for the penultimate step, we counted the number of $n$'s in the support of this delta function (similar to the result for $\# S_l$ in \eqref{latticepoints}): approximately $2^j$.

\subsubsection{Case II: $ \abs{n} \gtrsim \abs{m} $}
We decompose the sums from \ref{eq:delta2} in the following way:
\begin{equation*}
 \sum_{\substack{n,m \in \bZ^2 \\ \abs{n} \gtrsim \abs{m} }} \delta(***) = \sum_{j \geq i \geq 0} \sum_{\abs{m} \sim 2^i} \sum_{\abs{n} \sim 2^j} \delta(***) = \sum_{j \geq i \geq 0} \sum_{\abs{m} \sim 2^i} \# S_c,
\end{equation*}
where
\begin{equation}\label{Sdef}
 S_c = S_{m,j,\tau,p} := \{ \abs{n} \sim 2^j : \abs{p-n-m}^2 + \abs{n}^2 - \abs{m}^2 = -\tau \}.
\end{equation}

In order to compute $\# S_c$, we fix an $n_0 \in S_c$ and count the $\ell \in \bZ^2$ such that $n_0+\ell \in S_c$ also. Because $n_0+\ell \in S_c$ this equation must be satisfied:
\begin{equation} \begin{split}
 -\tau  &= \abs{p-(n_0 + \ell)-m}^2 + \abs{n_0 + \ell}^2 - \abs{m}^2 \\
& = \abs{p-n_0-m}^2 - 2\ell \cdot (p-2n_0-m) + 2 \abs{\ell}^2 + \abs{n_0}^2 - \abs{m}^2. \label{tosub}
\end{split} \end{equation}
Then we subtract from this the equation that $n_0$ must satisfy:
\begin{equation*}
 -\tau  = \abs{p-n_0-m}^2 + \abs{n_0}^2 - \abs{m}^2.
\end{equation*}
We arrive at the following equation for $\ell$:
\begin{equation}\label{elleqn}
 \abs{\ell}^2 + \ell \cdot (-p+2n_0+m) = 0.
\end{equation}
Considering $\bx  \in \bR^2$ instead of $\ell \in \bZ^2$, we examine $f$ defined by:
\begin{equation}\label{xeqn}
f(\bx) = \abs{\bx}^2 + \bx \cdot (-p+2n_0+m).
\end{equation}
For fixed $p,n_0,m$, the graph of $f$ is a paraboloid, whose level sets contain the solutions of \eqref{elleqn}, i.e., the lattice points. We are interested in the level set of $f$ at height 0, which is either empty or a circle $T$ satisfying the following equation (here we write $\bx = (x,y) \in \bR^2$ and $p-2n_0-m = (a,b)$ and complete the square):
\begin{equation}\label{circle}
 \biggl(x - \frac{a}{2}\biggr)^2 + \biggl(y - \frac{b}{2}\biggr)^2 = \frac{a^2 + b^2}{4}.
\end{equation}

An example of $T$ is depicted in Figure \ref{CzeroRlarge}, along with a disc $D := D(0, 2^j)$ centered at the origin representing the allowed range of $\abs{\ell}$, which comes from the facts that $\abs{n_0} \sim 2^j \sim \abs{n_0 + \ell}$.

The goal is to count the number of lattice points $\ell$ on $T$ and inside $D$. In examining different combinations of the parameters involved ($p,n_0,m$), there are two difficulties, which are handled as in \cite{DPST}. When $\abs{p}$ is small, $T$ is contained in $D$, and we must count all of the lattice points on $T$. This requires a number-theoretic estimate depending on the radius of $T$, from \cite{BP}. But this estimate would blow up for $\abs{p}$ large (e.g., case IIA), so there we must use the fact that the arc of $T$ contained in $D$ is relatively short and hence has only a few lattice points on it.

\begin{figure}[htbp]
\centering
\epsfig{file=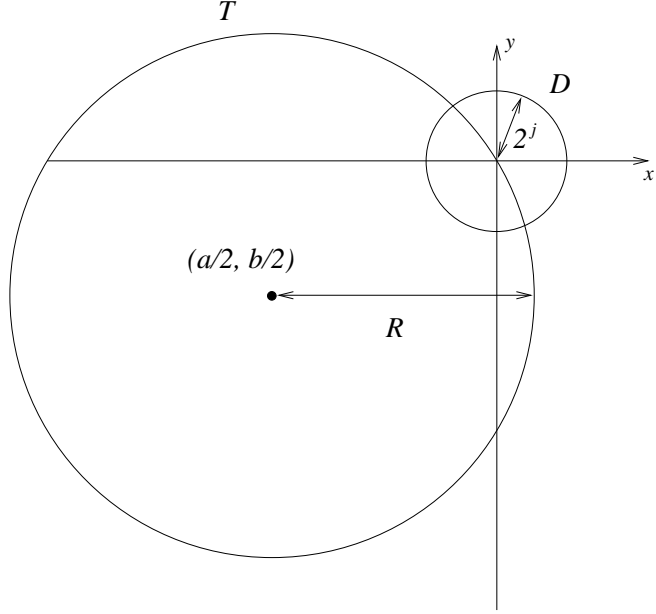, height=3.2in}
\caption{ For $R$ large enough, the arc $T \cap D$ contains only 2 lattice points (Case IIA). }
\label{CzeroRlarge}
\end{figure}

\textbf{Case IIA}: If $\abs{n} \gg \abs{m}$ and $ \abs{p} \gtrsim 2^{3j}$, then there are at most two lattice points on the arc $\gamma := T \cap D$. The proof follows from Lemma 4.4 in \cite{DPST} and requires $\gamma$ to satisfy:
\begin{equation}\label{gamma}
 \abs{\gamma} \lesssim R^{1/3}.
\end{equation}

To see that $\abs{p} \gtrsim 2^{3j}$ is sufficient to obtain this bound, we use a small-angle approximation for the left-hand side of \eqref{gamma} (the arclength is comparable to the diameter of $D$) to get:
\begin{equation} \begin{split}
 \abs{\gamma} \sim 2^j & \lesssim \abs{p}^{1/3} \\
& \lesssim \abs{p-m-2n_0}^{1/3} \\
& \sim R^{1/3}.
\end{split} \end{equation}

Then we can compute:
\begin{equation} \label{one}\begin{split}
  \sum_{\substack{ n, m \in \bZ^2 \\ \abs{n} \gg \abs{m} \\ \abs{p} \gtrsim 2^{3j} } } \frac{ \delta ( *** ) \langle p \rangle ^2 }{ \langle{ p-n-m }\rangle^2 \langle n \rangle ^2 \langle m \rangle ^2 } & \lesssim \sum_{i < j} \sum_{\abs{m} \sim 2^i} \frac{ \# S_c }{ \langle n \rangle^2 \langle m \rangle ^2 } \\
 & \lesssim \sum_{j = 0}^{\infty} \sum_{i = 0}^{j-1} \sum_{\abs{m} \sim 2^i} \frac{ 3 }{ 2^{2j} 2^{2i} } \\
 & \lesssim \sum_{j = 0}^{\infty} \frac{ j }{ 2^{2j} } < \infty.
 \end{split} \end{equation}

\begin{figure}[htbp]
\centering
\epsfig{file=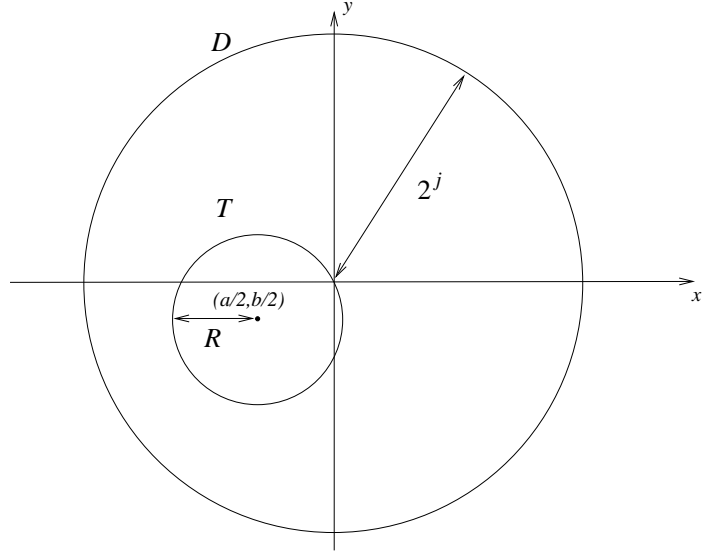, height=2.9in}
\caption{ For $R$ small compared to $2^j$, $T \cap D$ is the whole circle (Case IIB). }
\label{CposRsmall}
\end{figure}

\textbf{Case IIB}: If $\abs{n} \gg \abs{m}$ and $\abs{p} < 2^{3j}$, then because $p$ is small, the arc $\gamma$ is at least a large arc (perhaps the whole circle, as depicted in Figure \ref{CposRsmall}). So we bound the number of lattice points above by the number of lattice points on the whole circle $T$, using the Gauss lemma (see also \cite{BP}):
\begin{equation} \begin{split}
 \# \{ \ell \} \leq R^\eps & = \left( \abs{p-m-2n_0}^2 \right)^{\eps/2} \\
& \lesssim \left( (2^{3j})^2 \right)^{\eps/2} \\
& \sim 2^{3j \eps}. \label{wholebound}
\end{split} \end{equation}

\textbf{Subcase IIB1}: If $\abs{p} \gg \abs{n}$, then we use $\abs{p-n-m} \sim \abs{p}$ for the first step and  \eqref{wholebound} for the second:
\begin{equation} \begin{split}
  \sum_{\substack{ n, m \in \bZ^2 \\ 2^{3j} \gtrsim \abs{p} \gg \abs{n} \gg \abs{m} } } \frac{ \delta ( *** ) \langle p \rangle ^2 }{ \langle{ p-n-m }\rangle^2 \langle n \rangle ^2 \langle m \rangle ^2 } & \lesssim \sum_{i < j} \sum_{\abs{m} \sim 2^i} \frac{ \# S_c }{ \langle{ n }\rangle^2 \langle m \rangle ^2 } \\
 & \lesssim \sum_{i < j} \sum_{\abs{m} \sim 2^i} \frac{ 2^{ 3j\eps } }{ 2^{2j} 2^{2i} } \\
 & \lesssim \sum_{j = 0}^{\infty} \sum_{i = 0}^{j-1}  \frac{ 2^{ 3j\eps } }{ 2^{2j} } \\
 & \lesssim \sum_{j = 0}^{\infty} \frac{ j 2^{ 3j\eps } }{ 2^{2j} }.
\end{split} \end{equation}

\textbf{Subcase IIB2}: If $\abs{p} \lesssim \abs{n}$, and if $\abs{p-n-m} \gtrsim \abs{m}$, then we can do a change of variables, $z:=p- n-m$,
and we obtain the bound:
\begin{equation} \begin{split}
  \sum_{\substack{ \abs{n} \gg \abs{m} \\ \abs{p} \lesssim \abs{n} \\ \abs{p-n-m} \gtrsim \abs{m} } } & \frac{ \delta ( *** ) \langle p \rangle ^2 }{ \langle{ p-n-m }\rangle^2 \langle n \rangle ^2 \langle m \rangle ^2 }  \lesssim \sum_n \sum_m \frac{ \delta(***) }{ \langle{ p-m-n }\rangle^2 \langle m \rangle ^2 } \\
 & \lesssim \sum_z \sum_{|m|\lesssim |z|} \frac{ \delta(\abs{z}^2 + \abs{p-m-z}^2 - \abs{m}^2 + \tau) }{ \langle z \rangle^{2} \langle m \rangle^{2} } \lesssim
  \sum_{i < k} \sum_{\abs{m} \sim 2^i} \frac{ \# \tilde S_c }{ \langle{ z }\rangle^2 \langle m \rangle ^2 },
\end{split} \end{equation}
where
\begin{equation}\label{tildeSdef}
 \tilde S_c = \tilde S_{m,k,\tau,p} := \{ \abs{z} \sim 2^k : \abs{z}^2 + \abs{p-m-z}^2 - \abs{m}^2 = -\tau \}.
\end{equation}
By using the same arguments as above we deduce that if $|p|> 2^{3k}$ then we would obtain a bound as in \eqref{one}. On the other hand, in this case, $|p|\leq 2^{3k}$, so we obtain
$$\# \tilde S_c \lesssim 2^{\eps 3k}$$ which gives us an estimate similar to the one in \eqref{one}.


\textbf{Subcase IIB3}: If $\abs{p} \lesssim \abs{n}$, and if $\abs{p-n-m} \ll \abs{m}$, then we can do a change of variables $z:= n+m$, and proceed:
\begin{equation} \begin{split}
  \sum_{\substack{ \abs{n} \gg \abs{m} \\ \abs{p} \lesssim \abs{n} \\ \abs{p-n-m} \ll \abs{m} } } & \frac{ \delta ( *** ) \langle p \rangle ^2 }{ \langle{ p-n-m }\rangle^2 \langle n \rangle ^2 \langle m \rangle ^2 }  \lesssim \sum_n \sum_m \frac{ \delta(***) }{ \langle{ p-m-n }\rangle^2 \langle m \rangle ^2 } \\
 & \lesssim \sum_z \sum_m \frac{ \delta(\abs{p-z}^2 + \abs{z-m}^2 - \abs{m}^2 + \tau) }{ \langle p-z \rangle^{2} \langle m \rangle^{2} } \\
 & \lesssim \sum_z \sum_m \frac{ \delta(\abs{p-z}^2 + \abs{z-m}^2 - \abs{m}^2 + \tau) }{ \langle p-z \rangle^{2+\eps} \langle m \rangle^{2-\eps} } \\
 & \lesssim \sum_z \frac{1}{\langle p-z \rangle^{2+\eps}} \sum_m \frac{\delta(...)}{\langle m \rangle^{2-\eps}} < \infty. \\
\end{split} \end{equation}
In the last step, the sum in $m$ is one-dimensional due to the restraint, and $\langle p-z \rangle^{-\eps} > \langle m \rangle^{-\eps}$ allowed the borrowing to make the two-dimensional sum in $z$ converge.

\textbf{Case IIC}: $ \abs{n} \sim \abs{m} $

\textbf{Subcase IIC1}: If $\abs{n} \sim \abs{m} \ll 2^{3j} \lesssim \abs{p}$, then we use the fact that $\abs{p-n-m} \sim \abs{p}$ and cancel $\langle p \rangle^2$ with $\langle p-n-m \rangle^2$ to estimate like IIA:
\begin{equation} \begin{split}
  \sum_{\substack{ n, m \in \bZ^2 \\ \abs{n} \sim \abs{m} } } \frac{ \delta ( *** ) \langle p \rangle ^2 }{ \langle{ p-n-m }\rangle^2 \langle n \rangle ^2 \langle m \rangle ^2 } & \lesssim \sum_{i \sim j} \sum_{\abs{m} \sim 2^i} \frac{ \# S_c }{ \langle n \rangle^2 \langle m \rangle ^2 } \\
 & \lesssim \sum_{i \sim j} \sum_{\abs{m} \sim 2^i} \frac{ 3 }{ 2^{2i} 2^{2i} } \\
 & \lesssim \sum_{ i = 0 }^{\infty} 2^{-2i} < \infty.
\end{split} \end{equation}

\textbf{Subcase IIC2}: If $\abs{n} \sim \abs{m} \ll \abs{p} < 2^{3j}$, then we use $\abs{p-n-m} \sim \abs{p}$ for the initial cancellation and use the full-circle bound on the number of lattice points, \eqref{wholebound}, like IIB:
\begin{equation} \begin{split}
  \sum_{\substack{ n, m \in \bZ^2 \\ \abs{n} \sim \abs{m} } } \frac{ \delta ( *** ) \langle p \rangle ^2 }{ \langle{ p-n-m }\rangle^2 \langle n \rangle ^2 \langle m \rangle ^2 } & \lesssim \sum_{i \sim j} \sum_{\abs{m} \sim 2^i} \frac{ \# S_c }{ \langle{ n }\rangle^2 \langle m \rangle ^2 } \\
 & \lesssim \sum_{i \sim j} \sum_{\abs{m} \sim 2^i} \frac{ 2^{ 3j\eps } }{ 2^{2j} 2^{2i} } \\
 & \lesssim \sum_{j = 0 }^{\infty} \frac{ 2^{ 3j \eps } }{ 2^{2j} }.
 \end{split} \end{equation}

\textbf{Subcase IIC3}: If $\abs{p} \lesssim \abs{n} \sim \abs{m}$, then we cancel $\langle p \rangle^2$ with $\langle n \rangle^2$ and then use the change of variables $z:= n+m$ in order to bound it:
\begin{equation} \begin{split}
  \sum_{\substack{ n, m \in \bZ^2 \\ \abs{n} \sim \abs{m} } } & \frac{ \delta ( *** ) \langle p \rangle ^2 }{ \langle{ p-n-m }\rangle^2 \langle n \rangle ^2 \langle m \rangle ^2 } \\
& \lesssim \sum_{i \sim j} \sum_{\abs{m} \sim 2^i} \sum_{\abs{n} \sim 2^i} \frac{ \delta ( \tau + \abs{ p-n-m }^2 + \abs{ n }^2 - \abs{ m }^2 ) }{ \langle{ p-n-m }\rangle^2 \langle m \rangle ^2 } \\
& \lesssim \sum_{i \sim j} \sum_{\abs{m} \sim 2^i} \sum_{\abs{z} \lesssim 2^i} \frac{ \delta ( \tau +  \abs{ p-z }^2 +\abs{ z-m }^2 - \abs{ m }^2 ) }{ \langle{ p-z }\rangle^{2+\epsilon} \langle m \rangle ^{2-\eps} } \\
& \lesssim \sum_{i= 0 }^{\infty} \sum_{\abs{m} \sim 2^i}  \frac{ \# \tilde S_l }{  \langle m \rangle ^{2-\eps} }  \lesssim \sum_{i= 0 }^{\infty} \frac{1}{2^{i} }.
\end{split} \end{equation}
Here we used the definition, which, we already noticed, defines a line segment:
$$\tilde S_l=S_{z,i,\tau,p} := \{ \abs{m} \sim 2^i :   \abs{ p-z }^2 +\abs{ z-m }^2 - \abs{ m }^2 = -\tau  \},$$

This concludes the proof of (\ref{eq:delta2}), which is the case $k=1$; the proof for $k>1$ is similar.
\end{proof}

\appendix

\section{Sobolev and Poincar{\'e} type inequalities}
\setcounter{equation}{0}

The following lemma is a simple application of Sobolev inequalities.
\begin{lemma}\label{lm:sob}
For every $1< p \leq \infty$ there exists a constant $C_p$ such that
\[ \left| \langle \psi , V(x) \psi \rangle \right| = \left| \int_{\Lambda} \rd x \, V(x) |\psi (x)|^2 \right| \leq C_p \| V \|_p \langle \psi, (1-\Delta) \psi \rangle \]
for every $\psi \in L^2 (\Lambda)$. Moreover
\[ \left| \langle \psi, V(x_1 -x_2) \psi \rangle \right| \leq C \| V \|_1 \langle \psi, (1-\Delta_1)(1-\Delta_2) \psi \rangle \] for all $\psi \in L^2_s (\Lambda \times \Lambda, \rd x_1 \rd x_2)$.
\end{lemma}

To compare the potential $N^{2\beta} V (N^{\beta} x)$ with the limiting $\delta$-function, we use the following Poincar{\'e} type inequality.
\begin{lemma}\label{lm:poincare}
Suppose that $h \in L^1 (\Lambda)$ is a probability measure such that $\int_{\Lambda} \rd x \, (1+x^2)^{1/2} \, h(x) < \infty$; let $h_{\alpha} (x) = \alpha^{-2} h (x/\alpha)$. Then, for every $0\leq \kappa < 1$, there exists $C>0$ such that
\[ \Big| \tr \; J^{(k)} \left( h_{\alpha} (x_j - x_{k+1}) - \delta (x_j -x_{k+1}) \right) \gamma^{(k+1)} \Big| \leq C \alpha^{\kappa} \, \tri J^{(k)} \tri \, \tr \; \Big| S_j S_{k+1} \gamma^{(k+1)} S_{k+1} S_j \Big| \] for all non-negative $\gamma^{(k+1)} \in \cL^1_{k+1}$.
\end{lemma}
\begin{proof}
We prove the lemma in the case $k=1$. For $k > 1$ the proof is analogous. We decompose $\gamma^{(2)} = \sum_j \lambda_j |\ph_j \rangle \langle \ph_j|$ for $\ph_j \in L^2 (\Lambda^2)$, and for eigenvalues $\lambda_j \geq 0$. Then we have
\begin{equation}\label{eq:poin1}
\begin{split}
\tr \, J^{(1)} \left(h_{\alpha} (x_1 -x_2) - \delta (x_1 -x_2)\right) \gamma^{(2)}  =\; &\sum_{j} \lambda_j \langle \ph_j, J^{(1)} (h_{\alpha} (x_1-x_2) -\delta (x_1 -x_2)) \ph_j
\rangle  \\ = \; &\sum_{j} \lambda_j \langle \psi_j, (h_{\alpha} (x_1-x_2) -\delta (x_1 -x_2)) \ph_j \rangle
\end{split}
\end{equation}
where we defined $\psi_j = (J^{(1)} \otimes 1) \ph_j$. Next, switching to Fourier space, we observe that
\begin{equation}
\begin{split}
\langle \psi_j, &(h_{\alpha} (x_1-x_2) -\delta (x_1 -x_2)) \ph_j \rangle\\ & = \int \rd p_1 \rd p_2 \rd q_1 \rd q_2 \rd x \, \overline{\widehat{\psi}}_j (p_1, p_2) \widehat{\ph}_j (q_1,q_2) \, V(x) \left( e^{i \alpha x \cdot (p_1 -q_1)} - 1 \right) \, \delta (p_1 + p_2 - q_1 -q_2)
\end{split}\end{equation}
and thus, taking absolute value, we have, for arbitrary $0<\kappa<1$,
\begin{equation}\label{eq:poin2}
\begin{split}
\Big| \langle \psi_j, &(h_{\alpha} (x_1-x_2) -\delta (x_1 -x_2)) \ph_j \rangle\Big| \\ &\leq \alpha^{\kappa} \, \left( \int\rd x \, V(x) |x|^{\kappa} \right) \int \rd p_1 \rd p_2 \rd q_1 \rd q_2 \, |p_1 - q_1|^{\kappa} |\widehat{\psi}_j (p_1, p_2)| |\widehat{\ph}_j (q_1,q_2)| \delta (p_1 + p_2 - q_1 -q_2)
\end{split}
\end{equation}
Estimating $|p_1 -q_1|^{\kappa} \leq |p_1|^{\kappa} + |q_1|^{\kappa}$ we have to control two terms. We show how to control the term containing $|p_1|^{\kappa}$ (the term with $|q_1|^{\kappa}$ can be bounded similarly):
\begin{equation}
\begin{split}
\int \rd p_1 \rd p_2 \rd q_1 &\rd q_2 \, |p_1|^{\kappa} |\widehat{\psi}_j (p_1, p_2)| |\widehat{\ph}_j (q_1,q_2)| \delta (p_1 + p_2 - q_1 -q_2) \\ \leq \; &\int \rd p_1 \rd p_2 \rd q_1 \rd q_2 \, \delta (p_1 + p_2 - q_1 -q_2) \\ &\hspace{.5cm} \times \frac{(1+p^2_1)^{1/2} (1+p^2_2)^{1/2}}{(1+q_2^2)^{1/2} (1+q_1^2)^{1/2}} \, |\widehat{\psi}_j (p_1, p_2)|  \frac{(1+q_2^2)^{1/2} (1+q_1^2)^{1/2}}{(1+p^2_1)^{(1-\kappa)/2} (1+p^2_2)^{1/2}} |\widehat{\ph}_j (q_1,q_2)|  \\
\leq \; & \delta \int \rd p_1 \rd p_2 \rd q_1 \rd q_2 \, \frac{(1+p^2_1) (1+p^2_2)}{(1+q_2^2)(1+q_1^2)} \, |\widehat{\psi}_j (p_1, p_2)|^2  \delta (p_1 + p_2 - q_1 -q_2) \\
&+ \delta^{-1} \int \rd p_1 \rd p_2 \rd q_1 \rd q_2 \, \frac{(1+q_2^2) (1+q_1^2)}{(1+p^2_1)^{(1-\kappa)} (1+p^2_2)} |\widehat{\ph}_j (q_1,q_2)|^2 \delta (p_1 + p_2 - q_1 -q_2)\\
\leq \; & \delta \, \langle \psi_j, S_1^2 S_2^2 \psi_j \rangle  \sup_{p\in \Lambda} \int \rd q \, \frac{1}{(1+(p-q)^2) (1+q^2)}
\\&+ \delta^{-1} \, \langle \ph_j, S_1^2 S_2^2 \ph_j \rangle \sup_{q\in \Lambda} \int \rd p \, \frac{1}{(1+(q-p)^2)^{(1-\kappa)} (1+p^2)}
\end{split}
\end{equation}
for arbitrary $\delta >0$. Since \[ \sup_{q\in \Lambda} \int \rd p \, \frac{1}{(1+(q-p)^2)^{(1-\kappa)} (1+p^2)}  < \infty \] for all $0 \leq \kappa <1$, from (\ref{eq:poin1}) and (\ref{eq:poin2}) we obtain that
\begin{equation}
\begin{split}
\Big| \tr \, J^{(1)} \Big(h_{\alpha} (x_1 -x_2) - &\delta (x_1 -x_2)\Big) \gamma^{(2)} \Big| \\
\leq \; &C \alpha^{\kappa} \left( \delta \, \tr \, J^{(1)} S_1^2 S_2^2 J^{(1)} \gamma^{(2)} + \delta^{-1} \, \tr\, S_1^2 S_2^2 \gamma^{(2)} \right) \\
\leq \; &C \alpha^{\kappa} \left( \delta \, \tr\, S_1^{-1} J^{(1)} S_1^2 J^{(1)} S_1^{-1} S_1 S_2 \gamma^{(2)} S_2 S_1 + \delta^{-1} \, \tr \, S_1^2 S_2^2 \gamma^{(2)} \right)\\
\leq \; &C \alpha^{\kappa} \left( \delta \, \| S_1^{-1} J^{(1)} S_1 \| \, \| S_1 J^{(1)} S_1^{-1} \| + \delta^{-1} \right) \tr S_1^2 S_2^2 \gamma^{(2)} \\
\leq \; &C \alpha^{\kappa} \tri J^{(1)} \tri \, \tr S_1^2 S_2^2 \gamma^{(2)}
\end{split}
\end{equation}
where, in the last inequality we chose $\delta = \tri J^{(1)} \tri^{-1}$.
\end{proof}

\section*{Acknowledgements} K.~Kirkpatrick (kay@math.mit.edu) is supported by NSF postdoctoral research fellowship DMS-0703618. B.~Schlein (schlein@math.lmu.de) is on leave from Cambridge University; his research is supported by a Sofja Kovalevskaya Award of the Alexander von Humboldt Foundation. G.~Staffilani (gigliola@math.mit.edu) is partially supported by NSF grant DMS-0602678.

\end{document}